\documentclass[journal]{IEEEtran}

\usepackage{stmaryrd}
\usepackage{xspace}

\usepackage{subfigure}
\usepackage{url}
\usepackage{cite}
\usepackage{amsmath}
\usepackage{amssymb}
\usepackage{amsthm}
\usepackage{bbold, amsmath}
\usepackage[font={small}]{caption}

\usepackage{tikz}
\usepackage{enumitem}

\newtheorem{proposition}{Proposition}
\newtheorem{definition}{Definition}

\newtheorem{lemma}{Lemma}

\usepackage{mathtools}
\DeclarePairedDelimiter\ceil{\lceil}{\rceil}
\DeclarePairedDelimiter\floor{\lfloor}{\rfloor}

\DeclareMathOperator*{\argmax}{arg\,max}

\title{A Dynamic Game Approach to Strategic Design of Secure and Resilient Infrastructure Network}

\author{Juntao Chen,~\IEEEmembership{Student Member,~IEEE}, Corinne Touati, and Quanyan Zhu,~\IEEEmembership{Member,~IEEE} % <-this % stops a space
\thanks{This research is partial supported by grants SES-1541164 and ECCS-1847056 from National Science Foundation (NSF), and by award 2015-ST-061-CIRC01, U. S. Department of Homeland Security.}%

\thanks{J. Chen and Q. Zhu are with the Department of Electrical and Computer Engineering, Tandon School of Engineering, New York University, Brooklyn, NY 11201, USA. email: \{jc6412,qz494\}@nyu.edu}% <-this % stops a space
\thanks{
         C. Touati is with INRIA, F38330 Montbonnot Saint-Martin, France. email: corinne.touati@inria.fr}%
}

\begin{document}

\maketitle

\begin{abstract}
Infrastructure networks are vulnerable to both cyber and physical attacks. Building a secure and resilient networked system is essential for providing reliable and dependable services. To this end, we establish a two-player three-stage game framework to capture the dynamics in the infrastructure protection and recovery phases. Specifically, the goal of the infrastructure network designer is to keep the network connected before and after the attack, while the adversary aims to disconnect the network by compromising a set of links. With costs for creating and removing links, the two players aim to maximize their utilities while minimizing the costs. In this paper, we use the concept of subgame perfect equilibrium (SPE) to characterize the optimal strategies of the network defender and attacker. We derive the SPE explicitly in terms of system parameters. We further investigate the resilience planning of the defender and the strategic timing of attack of the adversary. Finally, we use case studies of UAV-enabled communication networks for disaster recovery to corroborate the obtained analytical results.
\end{abstract}
%\end{frontmatter}

\section{Introduction}
Infrastructure networks are increasingly connected due to the integration of the information and communications technologies (ICTs). For example, the introduction of smart meters has enabled the communications between the users and the utility companies. Communications with roadside units in vehicular networks can provide safety warnings and traffic information.
However, infrastructure networks are vulnerable to not only physical attacks (e.g., terrorism, theft or vandalisms) but also cyber attacks. These attacks can damage the connectivity of the infrastructure system and thus results in the performance degradation and operational dysfunction.  For instance, an adversary can attack the road sensor units and create traffic congestion \cite{al2012survey}. As a result, the transportation system can break down due to the loss of roads. An adversary can also launch denial-of-service attacks to disconnect communication networks \cite{pelechrinis2011denial}, resulting in inaccessibility of relevant database for air travel or financial transactions.

The cyber-physical nature of the infrastructure can also enable the coordinated attacks on the infrastructure systems that allow an adversary to use both cyber and physical approaches to disconnect networks. Therefore, infrastructure protection plays a significant role to maintain the connectivity of the infrastructure networks. One way to protect the network is to create redundant links in the network so that networks can be still connected despite arbitrary removal of links.
This approach has been used in traffic networks by creating multiple modes of transportation, in communication networks by adding extra wired or wireless links, and in supply chain networks by making orders from multiple suppliers.

Adding link redundancy is an effective approach when there is no knowledge of the target of the attacker, and thus the objective of the network designer is to secure the network by making the network robust to arbitrary removal of a fixed number of links.
However, it becomes expensive and sometimes prohibitive when the cost for creating links is costly, and the attacker is powerful.
Therefore, a paradigm shift to emphasize the recovery and response to attacks is critical, and the infrastructure resilience becomes essential for developing post-attack mechanisms to mitigate the impacts. Recovering the network from attack is a top priority for designers especially in the service-oriented critical infrastructures including electric power and communication networks \cite{DOEReport}.
With a limited budget of resources, it is essential to develop an optimal post-attack healing mechanism as well as a pre-attack secure mechanism holistically and understand the fundamental tradeoffs between security and resilience in the infrastructures.

To this end, we establish a two-player dynamic three-stage network game formation problem in which the infrastructure network designer aims to keep the network connected before and after the attack, while the objective of the adversary is to keep the network disconnected after the attack. Note that each player has a cost on creating or removing links.
Specifically, at the first stage of the game, the infrastructure network designer first creates a network with necessary redundancies by anticipating the impact of adversarial behavior. Then, an adversary attacks at the second stage by removing a minimum number of links of the network. At the last stage of the game, the network designer can recover the network after the attack by adding extra links to the attacked network. 

The resilience of the network is characterized by the capability of the network to maintain connectivity after the attack and the time it takes to heal the network. The security of the infrastructure is characterized by the capability of the network to withstand the attack before healing. 
Adding a large number of redundancies to the network can prevent the attack from disconnecting the network, but this approach can be costly. Hence, it is important to make strategic decisions and planning to yield a protection and recovery mechanism for the infrastructure with a minimum cost.

We adopt subgame perfect Nash equilibrium (SPE) as the solution concept of the dynamic game. We observe that with sufficient capabilities of recovery, the infrastructure can mitigate the threats by reducing the incentives of the attackers. We analyze SPE of the game by investigating two different parameter regimes. Further, we develop an optimal post-attack network healing strategy to recover the infrastructure network.  When an attacker is powerful (attack cost is low), we observe that the defender needs to allocate more resources in securing the network to reduce the incentives of the attacker. In addition, agile resilience and fast response to attacks are crucial in mitigating the cyber threats in the infrastructures. 

In the infrastructure network, agile resilience requires more effort of the network designer. Thus, when taking the resilience cost into account, the designer selects a mechanism including the defense and recovery strategies as well as the resilience ability jointly that yields the best net payoff. The attacker can also be strategic in choosing its attacking time. We find that when the defender does not recover the network, the attacker prefers to attack in an early phase and receives the total rewards afterward. In contrast, the attacker chooses to compromise the network at a later phase (though he does not really attack since the network is not connected initially), extracting all the utility from the initial time until the attacking phase. We finally use case studies on communication networks recovery based on unmanned aerial vehicles (UAVs) to illustrate our obtained theoretical results.

The contributions of this paper are summarized as follows.
\begin{enumerate}
\item[1)] We establish a two-player three-stage dynamic game framework to study the secure and resilient infrastructure network design. By considering the costs for creating and removing links, the network defender aims to keep the network connected while otherwise for the attacker.
\item[2)] We provide a complete analysis of the subgame perfect Nash equilibrium of the dynamic game which includes the defense and recovery strategies of the network defender and the attacking strategy of the adversary.
\item[3)] We derive constructive results on the resilience planning which specifies the optimal response time to attacks  for the defender as well as the strategic timing of attack that determines when to compromise for the adversary.
\end{enumerate}

\subsection{Related Work}
Communication network connectivity plays an important role in information exchange in various scenarios including civilian and military applications. To enhance the network connectivity against attacks, a number of methods have been proposed including two-way cooperative network formation \cite{chen2012joint}, secrecy graph approach \cite{goel2011effect}, and $q$-composite scheme \cite{zhao2017resilience}. Our work aims to improve the network connectivity by strategically investing link resources.

Security is a critical concern for infrastructure networks \cite{lewis2014critical,alpcan2010network,chen2017security}. In \cite{brown2006defending}, the authors have used bilevel and trilevel optimization models to design secure critical infrastructure against terrorist attacks. \cite{ten2010cybersecurity} has provided a comprehensive survey on cyber security of critical infrastructures and evaluated the adversarial impact using an attack-tree-based methodology. In \cite{ferdowsi2017colonel}, the authors have investigated secure state estimation of interdependent critical infrastructures through proposing a Colonel Blotto game framework and captured the dynamics of various components holistically using a novel integrated state-space model.
 A cross-layer design approach has been proposed in \cite{zhu2015game,pawlick2019istrict} to optimize the performance of cyber-physical control systems where the security is modeled using a game-theoretic framework. To further enhance the system performance, the strategy designed by the network operator should take the cascading failure effects into account due to the couplings between distinct network components \cite{crucitti2004model,wang2008attack,li2012cascading}. Cascading failures over networks have been widely studied in the literature. The authors in \cite{duenas2009cascading} have shown that topological changes are needed to increase cascading robustness, and improvements in network component tolerance alone do not ensure system robustness against cascading failures. In \cite{ash2007optimizing}, the authors have proposed an evolutionary algorithm to improve the network performance to cascading failures, and showed that clustering, modularity, and long path lengths are critical in designing robust large-scale infrastructure. Furthermore, \cite{zhu2012dynamic} has proposed a dynamic game-theoretic approach to investigate the coupling between cyber security policy and robust control design of industrial control systems under cascading failures. In addition, \cite{liao2017cascading} and \cite{huang2017large} have designed protective strategies using stochastic games for energy systems under cascading failures due to attacks. The authors in \cite{chen2019-TIFS} have developed strategic security investment strategies in IoT networks by capturing bounded rationality of players due to cognitive constraints. Different with previous literature on analyzing network failures using game approaches, our work captures the sequential move of attacker and defender and models the network structure explicitly. Furthermore, by leveraging dynamic games, graph theory and optimization, we provide a complete equilibrium analysis of the problem by considering network security and resilience jointly which is not a focus in previous works.  

In addition to the system security, resilience is another crucial property that needs to be considered by infrastructure network designers \cite{smith2011network}. In \cite{zhu2011robust}, the authors have proposed a hybrid  framework for robust and resilient control design with applications to power systems by considering both the unanticipated events and deterministic uncertainties. The authors in \cite{altman2016resilience} have studied the resilience aspect of routing problem in parallel link communication networks using a two-player game and designed stable algorithms to compute the equilibrium strategies. \cite{fang2016resilience} has studied the critical infrastructure resilience by focusing on two metrics, optimal repair time and resilience reduction worth, to measure the criticality of various components in the system.  The network resilience in our framework is quantified by the recovery time after the attack which needs to be strategically designed.

Dynamic game approaches have been widely used to investigate the network security and resilience. For example, \cite{shen2014differential} has used a differential game to model the malware defense in wireless sensor networks where the system designer chooses strategies to minimize the overall cost. A
stochastic repeated game and an iterative learning mechanism have been adopted for moving target defense in networks \cite{zhuhybrid}. In \cite{clark2012deceptive}, a multistage Stackelberg game has been studied for developing deceptive routing strategies for nodes in a multihop wireless communication network. Furthermore, \cite{chen2017heterogeneous} has proposed a three-player three-stage game-theoretic framework including two network operators and one attacker to enable the secure design of multi-layer infrastructure networks. Our framework is also a three-stage game but differs from \cite{chen2017heterogeneous} since we have one central network designer and take the system resilience into account.

The adopted method and framework in our infrastructure network design are relevant to the recent advances in adversarial networks \cite{goyal2014attack,dziubinski2013network,
bravard2017optimal,chen2019} and strategic network formation games \cite{chen2016resilient,bala2000noncooperative,chen2016interdependent}. 
Furthermore, the current work extends our previous one \cite{chen2017dynamic} in multiple aspects.  First, our goal in this work is to design the optimal protection, resilience planning and recovery strategies for infrastructure networks in a holistic manner which differs from \cite{chen2017dynamic} in which the critical resilience planning factor is not considered. Second, we investigate the new topic of network resilience and the strategic behavior of attacker in Section \ref{resilience}. Third, we provide the detailed proofs of all theoretical results which were omitted in \cite{chen2017dynamic}.  Fourth, we extensively expand the introduction and related work sections as well as the case studies section with more examples to explicitly illustrate the newly obtained analytical results.

\subsection{Organization of the Paper}
 The rest of the paper is organized as follows. Section \ref{formulation} formulates the problem. Dynamic game analysis are presented in Section \ref{analysis}. Section \ref{SPE_analysis} derives the SPE of the dynamic game. Network resilience and strategic timing of attack are investigated in Section \ref{resilience}. Case studies are given in Section \ref{case_study}, and Section \ref{conclusion} concludes the paper. 

\section{Dynamic Game Formulation}\label{formulation}

In this section, we consider an infrastructure system represented by a set ${\mathcal N}$ of $n$ nodes. The infrastructure designer can design a network with redundant links before the attack for protection and adding new links after the attack for recovery. Note that the attack action of the adversary can be enabled through cyber and physical approaches due to the integration of modern infrastructures with information and communication technologies. The sequence of the actions taken by the designer and the attacker is described as follows:
\begin{itemize}
\item[(i)] A Designer ($D$) aims to create a network between these nodes and protect it against a malicious attack;
\item[(ii)] After some time of operation, an Adversary ($A$) puts an attack on the network by removing a subset of its links;
\item[(iii)] Once the $D$ realizes that an attack has been conducted, it has the opportunity to heal its network by constructing new links (or reconstructing some destroyed ones).
\end{itemize}
In addition, the timing of the actions also play a significant role in determining the optimal strategies of both players. We normalize the horizon of the event from the start of the preparation of infrastructure protection to a time point of interest as the time internal $[0, 1]$. This normalization is motivated by the observation made in \cite{DOEReport} where the consequences of  fifteen major storms occurring between 2004 and 2012 are plotted over a normalized duration of the event. We let $\tau$ and $\tau_R$ represent, respectively, the fraction of time spent before the attack (system is fully operational) and between the attack and the healing phase. This is illustrated in Fig.~\ref{fig:time}. 

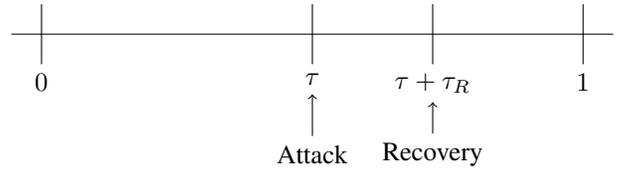
\begin{figure}[h]
\centering{
\begin{tikzpicture}[scale=0.8]
\draw (0,0) -- (10,0);
\draw (0.5, 0.5) -- (0.5,-0.5) node [anchor=north]{$0$};
\draw (5,0.5) -- (5,-0.5) node [anchor=north] (t) {$\tau$};
\draw (7,0.5) -- (7,-0.5) node [anchor=north] (tr) {$\tau+\tau_R$}; 
\draw (9.5,0.5) -- (9.5,-0.5) node [anchor=north]{$1$};
\node (A) at (5,-2) {Attack};
\node (R) at (7,-2) {Recovery};
\path[->] (A) edge node {} (t);
\path[->] (R) edge node {} (tr);
\end{tikzpicture}
}
\caption{Attack and defense time fractions. The attacker compromises the network at time $\tau$, and the defender recovers it after $\tau_R$ amount of time.}
\label{fig:time}
\end{figure}
The goal of the designer or the defender is to create protection and recovery mechanisms to keep its network operational, i.e., connected in this case. Let ${\mathcal E}_1$ be the set of links created by the defender initially (i.e., at time $0$). ${\mathcal E}_A \subseteq {\mathcal E}_1$ is the set of links removed (attacked) by the adversary and $\mathcal{E}_2$ is the set of links created by the defender after the attack (at fraction $\tau+\tau_R$ of the time horizon). Regardless of the time stamp, creating (resp. removing) links has a unitary cost $c_D$ (resp. $c_A$). The adversary aims to disconnect the network. Thus, for any set ${\mathcal E}$, we define $\mathbb{1}_E$ which equals $1$ if the graph $({\mathcal N},{\mathcal E})$ is connected and $0$ otherwise. Values of $\tau$, $\tau_R$, $c_A$ and $c_D$ are assumed as common knowledge to both $D$ and $A$ first, and later we investigate the strategic selections of $\tau$ and $\tau_R$. As a tie-breaker rule, if the output/utility is the same for $A$, then $A$ chooses to attack the network with the largest number of link removals. Similarly, $D$ chooses not to create links if its utility is the same.

\textit{Remark:} The link creation cost is treated as identical in the framework. Here, $c_D$ can capture various application scenarios. For example, in a large complex network with heterogeneous link costs, analyzing the strategy of $D$ becomes intractable. A viable choice for $D$ is to consider the mean link creation cost captured by $c_D$ which gives an approximation of the network. Another case is that $D$ considers the largest single link creation cost denoted by $c_D$, and thus it captures the worst case in which $D$ is conservative in designing the strategies. In sum, considering an identical $c_D$ is reasonable, and also it makes the technical analysis of the problem tractable.

The utility for the designer (resp. adversary) is equal to the fraction of time the network is connected (resp. disconnected) minus the costs of creating (resp. removing) the links. Hence, the payoff functions of the designer and the adversary are represented by $U_D$ and $U_A$, respectively, as follows:
$$
\begin{array}{@{}l@{}l@{}}
U_D({\mathcal E}_1, {\mathcal E}_2, {\mathcal E}_A) =  &(1-\tau-\tau_R) \mathbb{1}_{E_1\backslash E_A\cup E_2} +  \tau \mathbb{1}_{E_1} \\
 & \hfill + \tau_R \mathbb{1}_{E_1\backslash E_A} -c_D (|{\mathcal E}_1| + |{\mathcal E}_2|), \\
U_A({\mathcal E}_1, {\mathcal E}_2, {\mathcal E}_A) = & (1-\tau-\tau_R) (1-\mathbb{1}_{E_1\backslash E_A\cup E_2}) -c_A |{\mathcal E}_A| \\
 &\hfill +  \tau (1-\mathbb{1}_{E_1})+ \tau_R (1-\mathbb{1}_{E_1\backslash E_A}), 
\end{array}
$$
where $|\cdot|$ denotes the cardinality of a set. In addition, $\mathbb{1}_{E_1\backslash E_A\cup E_2}$ means that a network including $|\mathcal{N}|=n$ nodes contains a set ${\mathcal E}_1\backslash {\mathcal E}_A\cup {\mathcal E}_2$ of links. Note that if the fraction of time and the cost of links metrics cannot be directly added up in the utility functions, we can use a conversion factor to transform one metric to the other. Therefore, the formulated utility functions for $D$ and $A$ are still valid.

Since both players are strategic, we study the SPE and analyze the strategies of the players to the sets $({\mathcal E}_1, {\mathcal E}_A, {\mathcal E}_2)$. Thus, we seek triplets $({\mathcal E}_1, {\mathcal E}_A, {\mathcal E}_2)$ such that ${\mathcal E}_2$ is a best response to $({\mathcal E}_1,{\mathcal E}_A)$ and that given ${\mathcal E}_1$, $({\mathcal E}_A, {\mathcal E}_2)$ is also a SPE. In other words, the SPE involves the analysis of the following three sequentially nested problems starting from the last stage of the designer's recovery problem to the first stage of the designer's protection problem:
\begin{itemize}
\item[(i)] Given the strategies ${\mathcal E}_1$ and ${\mathcal E}_A$, player $D$ chooses\\ ${\mathcal E}_2^*({\mathcal E}_1, {\mathcal E}_A) \in \argmax_{{\mathcal E}_2} U_D ({\mathcal E}_1, {\mathcal E}_A, {\mathcal E}_2)$;
\item[(ii)] Given ${\mathcal E}_1$, the adversary chooses\\ ${\mathcal E}_A^*({\mathcal E}_1) \in \argmax_{{\mathcal E}_A} U_A ({\mathcal E}_1, {\mathcal E}_A, {\mathcal E}_2^*({\mathcal E}_1, {\mathcal E}_A))$;
\item[(iii)] Player $D$ chooses\\ ${\mathcal E}_1^* \in  \argmax_{{\mathcal E}_1}U_D ({\mathcal E}_1, {\mathcal E}_A^*({\mathcal E}_1), {\mathcal E}_2^*({\mathcal E}_1, {\mathcal E}_A))$.
\end{itemize}
The equilibrium solution $({\mathcal E}_1, {\mathcal E}_A, {\mathcal E}_2)$ that solves the above three problems consistently is an SPE of the two-player dynamic game.

\textit{Comments on the game formulation:} In the established model, the attacking time $\tau$ and attacker's cost $c_A$ are assumed to be known by $D$. More practically, $D$ may have no perfect information on the attacker's parameters, and only the distributions of $\tau$ and $c_A$ are available. Then, $D$ can calculate the expected values of $\tau$ and $c_A$. The analysis in the paper is still valid to design the defensive strategy of $D$ at time 0. However, $A$'s behavior may not be the same as expected by $D$ which leads to a random network after the attack. Thus, $D$ needs to determine the healing strategy $\mathcal{A}_2$ again at time $\tau$. This creates another layer of decision-making problem for $D$ which is an optimization problem itself instead of a game as $A$'s behavior has been revealed. Other than capturing the unknown parameters $\tau$ and $c_A$ through their expected values, we can also model the game by considering the incomplete information directly. This yields a formulation of  dynamic Bayesian game with a random type parameter including $\tau$ and $c_A$ which is nontrivial to solve.

\section{Dynamic Game Analysis}\label{analysis}
In this section, we analyze the possible configurations of the infrastructure network at SPE.

We first note that $c_A$ should be not too large, since otherwise $A$ cannot be a threat to $D$. Similarly, $c_D$ should be sufficiently small so that the $D$ can create a connected network:
\begin{lemma}
If $c_A > 1- \tau$, then $A$ has no incentive to attack any link. In addition,
if $c_D > \frac{1}{n-1}$, then $D$ has no incentive to create a connected network.
\end{lemma}
\begin{proof}
Suppose that $c_A > 1- \tau$. Let ${\mathcal E}_1$ be given and $\phi := \tau (1-\mathbb{1}_{E_1})$. If $A$ decides not to remove any link, then its payoff is $\phi + \tau_R (1-\mathbb{1}_{E_1})+(1-\tau-\tau_R) (1-\mathbb{1}_{E_1\cup E_2}) \geq \phi$. Otherwise, $|{\mathcal E}_A| \geq 1$ and 
$U_A({\mathcal E}_1, {\mathcal E}_2, {\mathcal E}_A) \leq  \phi + (1-\tau-\tau_R) (1-\mathbb{1}_{E_1\backslash E_A\cup E_2}) -c_A
+ \tau_R (1-\mathbb{1}_{E_1\backslash E_A}) \leq \phi + 1-\tau-c_A < \phi$. 
Thus, it is a best response for $A$ to play ${\mathcal E}_A = \emptyset$.
Similarly, if $c_D > \frac{1}{n-1}$, then if $D$ plays ${\mathcal E}_1 = {\mathcal E}_2 = \emptyset$, its utility is $0$. Otherwise, its utility is bounded above by $1-(n-1) c_D$ which corresponds to a connected tree network with the minimum number of links. 
\end{proof}

In the following, we thus suppose that $c_A < 1-\tau$ and $c_D < \frac{1}{n-1}$.

%------------------------------------------------------------------------------
%------------------------------------------------------------------------------
%------------------------------------------------------------------------------

Note that the SPE can correspond only to a set of situations summarized as follows.
\begin{lemma}\label{lem:situations}
Suppose that $({\mathcal E}_1, {\mathcal E}_A, {\mathcal E}_2)$ is an SPE. Then, we are necessarily in one of the situations given in Table~\ref{tab}.
\begin{table}[h]
$$\begin{array}{c|ccc}
\mathrm{Situation} & \mathbb{1}_{E_1} & \mathbb{1}_{E_1\backslash E_A} & \mathbb{1}_{E_1\backslash E_A\cup E_2}\\
\hline
1 & 1 & 1 & 1\\
2 & 1 & 0 & 1\\
3 & 1 & 0 & 0\\
4 & 0 & 0 & 1\\
5 & 0 & 0 & 0
\end{array}$$
\caption{Different potential combinations of values of $\mathbb{1}_{E_1}$, $\mathbb{1}_{E_1\backslash E_A}$ and $\mathbb{1}_{E_1\backslash E_A\cup E_2}$ at the SPE.}\label{tab}
\end{table}
\end{lemma}

\begin{proof}
Note that, in total, $8$ situations should be possible. However, if $ \mathbb{1}_{E_1} = 0$, then it is impossible that $\mathbb{1}_{E_1\backslash E_A} = 1$. Therefore, the situations where $(\mathbb{1}_{E_1}, \mathbb{1}_{E_1\backslash E_A}, \mathbb{1}_{E_1\backslash E_A\cup E_2})$ equaling to $(0,1,0)$ and $(0,1,1)$ are not possible. 
Further, if $\mathbb{1}_{E_1\backslash E_A}=1$, then it is impossible that $\mathbb{1}_{E_1\backslash E_A\cup E_2}=0$. Thus, the situation $(\mathbb{1}_{E_1}, \mathbb{1}_{E_1\backslash E_A}, \mathbb{1}_{E_1\backslash E_A\cup E_2})=(1,1,0)$ is impossible. All other combinations are summarized in Table~\ref{tab}. 
\end{proof}

In Situations 4 and 5, $D$ does not create a connected network in the beginning, and thus $A$ has no incentive to attack the network at phase $\tau$.
The structure of the SPE depends on the values of the parameters of the game. In particular, it depends on whether $D$ has incentive to fully reconstruct (heal) the system after the attack of $A$. More precisely, if $1-\tau-\tau_R > (n-1)c_D$, then $D$ prefers to heal the network even if all links have been compromised by the attacker. Otherwise, there should be a minimum number of links remained after the attack for the $D$ to heal the network at the SPE. We sequentially analyze these two cases in Sections~\ref{sec:1} and \ref{sec:2}, respectively.

\section{SPE Analysis of the Dynamic Game}\label{SPE_analysis}
Depending on the parameters, we derive SPE of the dynamic game in two regimes: $1-\tau-\tau_R > (n-1)c_D$ and the otherwise in this section.

Before presenting the results, we first present the definition of Harary network \cite{harary} which plays an essential role in the SPE analysis. For a network containing $n$ nodes being resistant to $k$ link attacks, one necessary condition is that each node should have a degree of at least $k+1$, yielding the total number of links more than $\ceil*{\frac{(k+1)n}{2}}$, where $\ceil*{\cdot}$ denotes the ceiling operator. Harary network presented below can achieve this lower bound on the number of required links.
\begin{definition}[Harary Network \cite{harary}]
In a network containing $n$ nodes, Harary network is the optimal design that uses the minimum number of links equaling $\ceil*{\frac{(k+1)n}{2}}$ for the network still being connected after removing any $k$ links. 
\end{definition}
The constructive method of general Harary network can be described with cycles as follows. It first creates the links between node $i$ and node $j$ such that $(|i-j|\mod n)=1$, and then $(|i-j|\mod n)=2$, etc. When the number of nodes is odd, then the last cycle of link creation is slightly different since $\frac{(k+1)n}{2}$ is not an integer. However, the bound $\ceil*{\frac{(k+1)n}{2}}$ can be still be achieved.  For clarity, we illustrate three cases in Fig. \ref{Harary_figs} with $n=5,\ 7$ under different security levels $k=2,\ 3$.

\begin{figure}[!t]
\begin{centering}
\includegraphics[width=0.9\columnwidth]{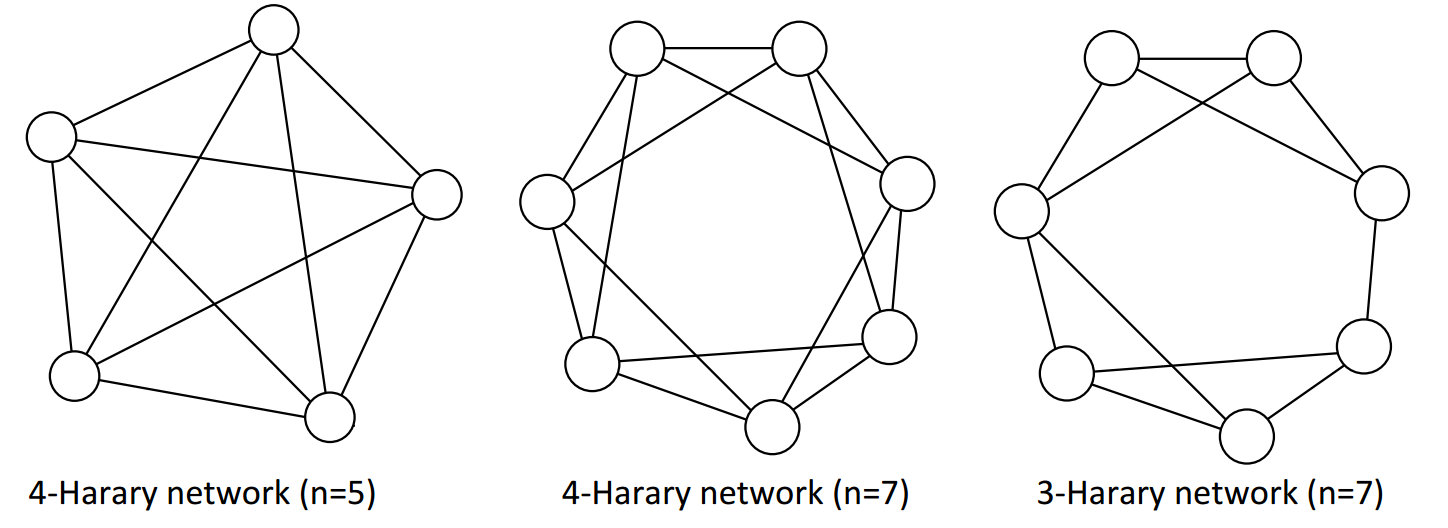} % double-c: 1; single-c:0.7
\par\end{centering}
\caption{\label{Harary_figs}
Illustration of Harary networks with different number of nodes and security levels.}
\end{figure}

Another critical network topology used in the analysis is the tree network defined as follows.
\begin{definition}[Tree network \cite{godsil2013algebraic}]
A tree is an undirected graph in which any two nodes are connected by exactly one path. Equivalently, the network is a tree if and only if it is connected and acyclic (contains no cycles).
\end{definition} 

\subsection{Regime 1: $1-\tau-\tau_R > (n-1)c_D$}
\label{sec:1}

In the case where $1-\tau-\tau_R > (n-1)c_D$, $D$ always reconstructs the network to be connected after the attack. The potential SPE can occur in only three of the Situations in Table~\ref{tab}, and we summarize them in the following proposition.
\begin{proposition}\label{prop:1}
Suppose that $1-\tau-\tau_R > (n-1)c_D$ and let $k_A^R \equiv \floor*{\frac{\tau_R}{c_A}}$, where $\floor*{\cdot}$ denotes the floor operator (resp. to the ceiling operator $\ceil*{\cdot}$). Note that $k_A^R$ is the largest number of links that $A$ can compromise to have a nonnegative payoff. Then, the SPE of the game is unique and satisfies:
\begin{itemize}
\item[(1)] If $\tau_R < c_A$, then $U_D = 1-(n-1)c_D$ and $U_A = 0$ (Situation $1$).
\item[(2)] Otherwise, i.e., $\tau_R \geq c_A$, and
\begin{itemize}
\item[(i)] if $\tau > c_D$ and $\tau_R > c_D \ceil*{\frac{n(k_A^R-1)}{2}}$ or if $\tau<c_D$ and $ \tau+\tau_R > c_D \ceil*{\frac{n(k_A^R-1)}{2}+1}$, then the SPE 
satisfies \\$\left\{\begin{array}{l}
U_D = 1-c_D \ceil*{\frac{n(k_A^R+1)}{2}} \\
U_A=0
\end{array}\right.$ (Situation $1$).
\item[(ii)] If $\tau > c_D$ and $\tau_R < c_D \ceil*{\frac{n(k_A^R-1)}{2}}$,  then the SPE satisfies $\left\{\begin{array}{l}
U_D = 1-\tau_R-n c_D \\ U_A = \tau_R-c_A \end{array}\right.$ (Situation $2$).
\item[(iii)] If $\tau < c_D$ and $ \tau+\tau_R < c_D \ceil*{\frac{n(k_A^R-1)}{2}+1}$, then the SPE 
satisfies\\
$\left\{\begin{array}{l}
U_D = 1-\tau-\tau_R-(n-1)c_D\\
U_A=\tau+\tau_R 
\end{array}\right.$ (Situation $4$).
\end{itemize}
\end{itemize}
\end{proposition}

Proposition \ref{prop:1} is a direct consequence of the following lemma. Note that the conditions in Proposition \ref{prop:1} are obtained via comparing $D$'s utility $U_D$ at various SPEs in Table~\ref{tab:SPE}.

\begin{lemma}\label{lem:case1}
Suppose that $1-\tau-\tau_R > (n-1)c_D$. The potential SPEs have the properties given in Table~\ref{tab:SPE}.
\begin{center}
\begin{table}[h]
{\small
%\hfill{}
$$\hspace{-1em}\begin{array}{|@{}c@{\,}|@{}c@{\;}c@{\;}c@{\,}|@{\,}c@{\;\;}c@{}|}
\mathrm{Situation} & |{\mathcal E}_1| & |{\mathcal E}_A| & |{\mathcal E}_2| & U_D & U_A \\
\hline
1 \& k_A^R>0 & \ceil*{\frac{n(k_A^R+1)}{2}} & 0 & 0 & 1-c_D \ceil*{\frac{n(k_A^R+1)}{2}}& 0\\
1 \& k_A^R=0 & n-1 & 0 & 0 & 1-(n-1) c_D & 0\\
2 & n-1 & 1 & 1 & 1- \tau_R-n c_D & \tau_R-c_A\\
4 & 0 & 0 & n-1 & 1-\tau-\tau_R-(n-1)c_D & \tau+\tau_R\\
\hline
\end{array}$$
}
\hfill{}
\caption{Different potential SPEs when $1-\tau-\tau_R > (n-1)c_D$ (Note: $k_A^R = \floor*{\frac{\tau_R}{c_A}}$).}
\label{tab:SPE}
\end{table}
\end{center}
\end{lemma}

\begin{proof}
First note that any connected network contains at least $n-1$ links. Conversely, any set of nodes can be made connected by using exactly $n-1$ links (any spanning tree is a solution). We consider a situation where $\mathbb{1}_{E_1\backslash E_A}=0$. Then, either $D$ decides not to heal the network and receives a utility of $U^* = \tau \mathbb{1}_{E_1} -c_D |{\mathcal E}_1|$, or it decides to heal it (by using at most $n-1$ links) and receives a utility of at least $\overline{U} = (1-\tau-\tau_R) +  \tau \mathbb{1}_{E_1} -c_D (|{\mathcal E}_1| + n-1)$. The difference is $\overline{U}-U^* =  (1-\tau-\tau_R) -c_D (n-1) > 0$. Thus, $D$ always prefers to heal the network after the attack of $A$. Therefore, Situations $3$ and $5$ contain no SPE. 

Next we consider Situation $4$. Since $\mathbb{1}_{E_1\backslash E_A\cup E_2}=1$, then $D$ needs to create in total at least $n-1$ links: $|{\mathcal E}_1|+|{\mathcal E}_2| \geq n-1$. Therefore, an optimal strategy is ${\mathcal E}_1 = \emptyset$ and $|{\mathcal E}_2| = n-1$. Since ${\mathcal E}_1 = \emptyset$, the optimal strategy of $A$ is ${\mathcal E}_A = \emptyset$.

In Situation $2$, $({\mathcal N}, {\mathcal E}_1)$ is connected, and thus $|{\mathcal E}_1| \geq n-1$. Further, $\mathbb{1}_{E_1}=1$ and $\mathbb{1}_{E_1\backslash E_A} = 0$, and thus $|{\mathcal E}_A| \geq 1$. 
Since $1-\tau-\tau_R >(n-1)c_D$, then $A$ should remove the minimum number of links to disconnect the network, and we obtain the result.

Finally, in Situation $1$, since $\mathbb{1}_{E_1\backslash E_A}=1$, then $D$ does not need to create any link during the healing phase: ${\mathcal E}_2= \emptyset$. Since $1-\tau-\tau_R > (n-1)c_D$, then $A$ attacks at most $k_A^R$ links if and only if it obtains a nonnegative reward, i.e., $k_A^R$ is the largest integer such that $\tau_R-c_Ak_A^R \geq 0$ which yields $k_A^R = \floor*{\frac{\tau_R}{c_A}}$. Thus, $D$ designs a network that is resistant to an attack compromising up to $k_A^R$ links. Such solution network is the ($|{\mathcal N}|,k_A^R+1$)-Harary network \cite{harary}.
\end{proof}

\textit{Examples:} For clarify, we depict the strategies of $D$ and $A$ at various SPEs using examples shown in Fig. \ref{SPE_ex1}. The network contains 5 nodes. Depending on the relationship between parameters shown in Proposition \ref{prop:1}, the game admits various SPEs. Four possible SPEs with specific actions taken by $D$ and $A$ are presented. For example, when the SPE lies in Situation 1 with $k_A^R = 3$, then at least 10 links are necessary for the network being resistant to 3 attacks. Therefore, $D$ creates a 4-Harary network initially in which each node has at least a degree of 4. In comparison, when $k_A^R = 0$ and the SPE is in Situation 1, then creating a connected tree network is sufficient for $D$ since $A$ is not capable to compromise any link. The SPEs corresponding to Situations 2 and 4 are shown in Figs. \ref{SPE_ex1_3} and \ref{SPE_ex1_4}, respectively.

\begin{figure}[!t]
  \centering
  \subfigure[Situation 1 and $k_A^R = 3$]{
    \includegraphics[width=0.9\columnwidth]{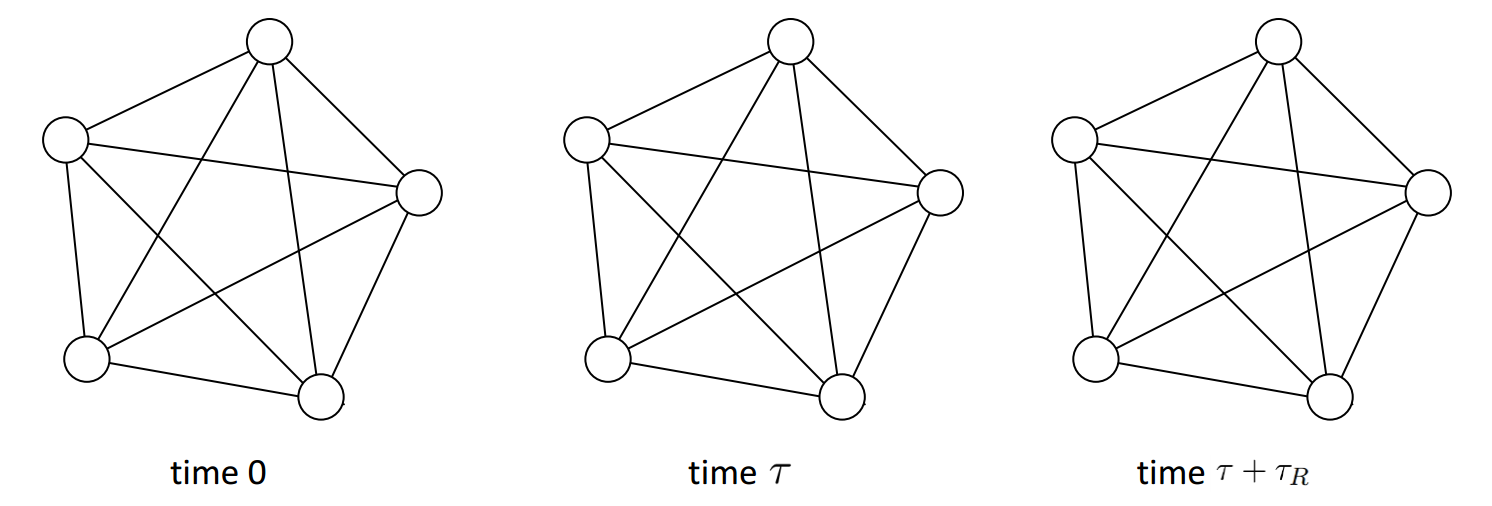}\label{SPE_ex1_1}}
	 \subfigure[Situation 1 and $k_A^R = 0$]{
    \includegraphics[width=0.9\columnwidth]{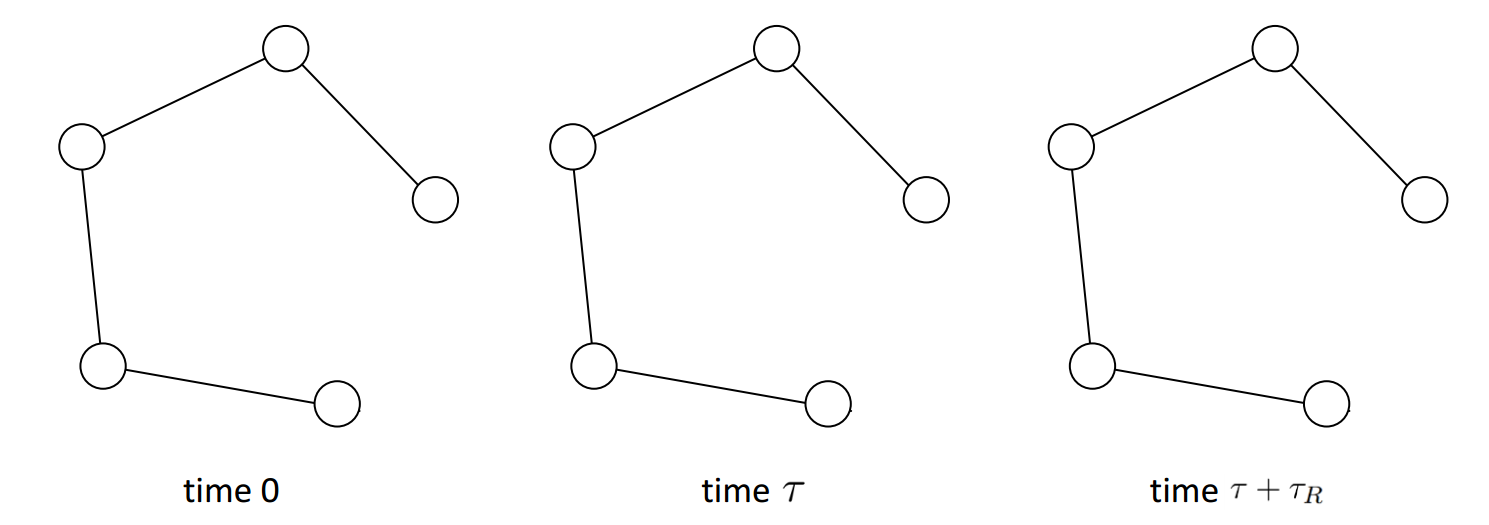}\label{SPE_ex1_2}}  
    	 \subfigure[Situation 2]{
    \includegraphics[width=0.9\columnwidth]{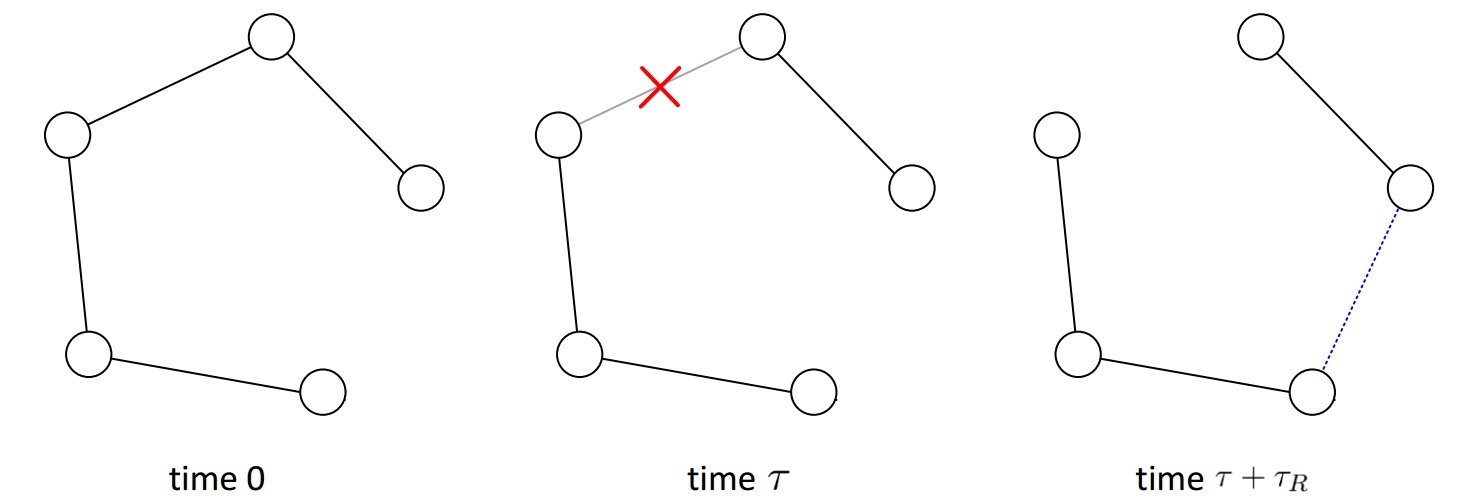}\label{SPE_ex1_3}}
    	 \subfigure[Situation 4]{
    \includegraphics[width=0.9\columnwidth]{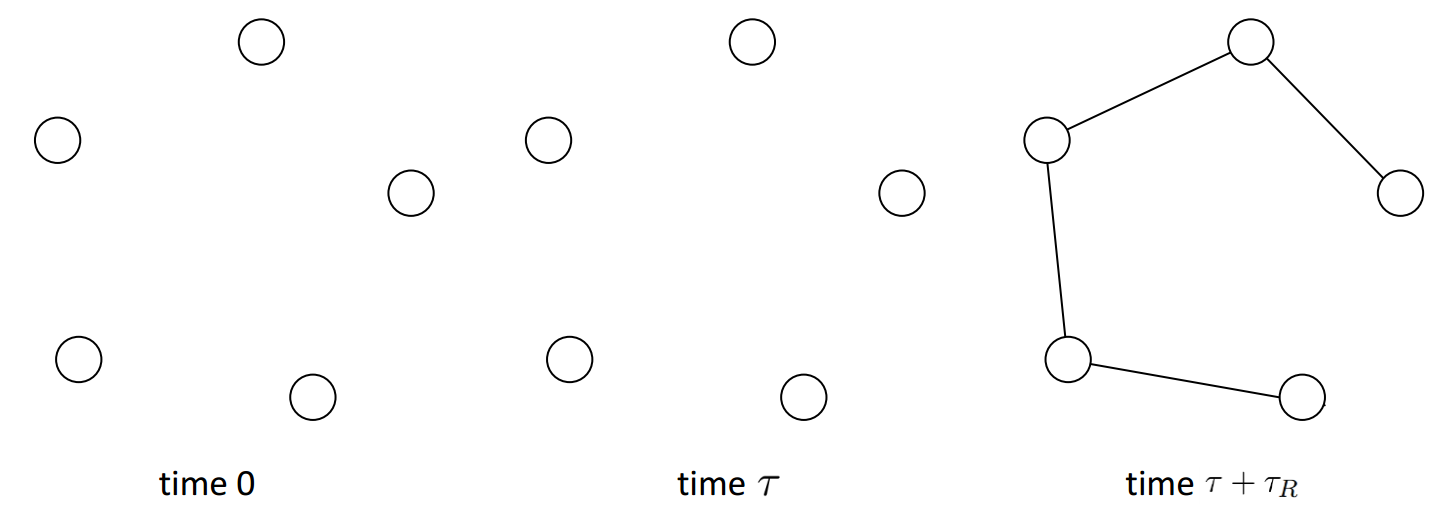}\label{SPE_ex1_4}}     
  \caption[]{Strategies of $D$ and $A$ at different SPEs in regime 1. The network contains 5 nodes. In (a), the SPE is in Situation 1 and $k_A^R = 3$. Thus, at least 10 links are necessary for it being resistant to 3 attacks. In (b), when $k_A^R = 0$ and the SPE lies in Situation 1, a tree network is created by the defender following no actions of $A$ and $D$. In (c), the SPE is in Situation 2, and $A$ will compromise any one link at time $\tau$ and $D$ will heal one link to reconnect the network. In (d), $D$ will not protect the network at time 0 but will connect the network at time $\tau+\tau_R$ which shows SPE in Situation 4.}
  \label{SPE_ex1}
\end{figure}

Based on Lemma \ref{lem:case1}, the stragies of two players at SPE in regime 1 are summarized as follows. Under Situation 1 and $k_A^R>0$, $A$ does not attack and $D$ creates a connected ($|{\mathcal N}|,k_A^R+1$)-Harary network at phase 0. Under Situation 1 and $k_A^R=0$, $D$ simply creates a connected network with the minimum number  of $n-1$ links which can be achieved by any tree-structured network, and $A$ admits a null strategy. In Situation 2,  $D$ initially constructs a tree network using $n-1$ links, and $A$ attacks any one link at phase $\tau$ followed by $D$ recovering the network at phase $\tau+\tau_R$. Finally, for Situation 4, $A$ does not attack, and $D$ constructs a connected tree network only at phase $\tau+\tau_R$.

\subsection{Regime 2: $1-\tau-\tau_R < (n-1) c_D$}
\label{sec:2}

We now consider the case where $D$ has an incentive, at phase $\tau+\tau_R$, to heal the network if at most $k$ links are required to reconnect it, where $k<n-1$ and 
\begin{equation}
k \equiv \floor*{\frac{1-\tau-\tau_R}{c_D}}.
\end{equation}

We sequentially study the potential SPE in Situations $3$, $4$ and $5$ in Lemma~\ref{lem:penible}, Situation $2$ in Lemma~\ref{lem:sit2}, and Situation $1$ in Lemma~\ref{lem:sit1}.

\begin{lemma}\label{lem:penible}
If $1-\tau-\tau_R < (n-1) c_D$, we have the following results:
\begin{itemize}
\item[(i)] Any SPE in Situation $3$ satisfies ${\mathcal E}_2 = \emptyset$, $|{\mathcal E}_A| = k+1$ and $|{\mathcal E}_1| = n-1$,   leading to utilities $U_D = \tau - (n-1)c_D$ and $U_A = 1-\tau-(k+1)c_A$ (occurs only if $\floor{\frac{1-\tau}{c_A}}\geq k+1$);
\item[(ii)] There exists no SPE in Situation $4$;
\item[(iii)] The only potential SPE in Situation $5$ is the null strategy: ${\mathcal E}_1 = {\mathcal E}_2 = {\mathcal E}_A = \emptyset$, leading to utilities $U_D = 0$ and $U_A = 1$.
\end{itemize}
\end{lemma}

\begin{proof}
Suppose that an SPE occurs in Situation $5$. Since the network is always disconnected, then $U_D = - c_D (|{\mathcal E}_1|+|{\mathcal E}_2|)$. The maximum utility is obtained when ${\mathcal E}_1 = {\mathcal E}_2 = \emptyset$. Thus, ${\mathcal E}_A= \emptyset$.

In Situation $4$, since any connected network contains at least $n-1$ links, then the maximum utility of $D$ is  $U_D({\mathcal E}_1, {\mathcal E}_2, {\mathcal E}_A) =  (1-\tau-\tau_R) - c_D (n-1)<0$. Thus, $D$ is better off with a null strategy (occurring in Situation $5$).

In Situation $3$, since $\mathbb{1}_{E_1} = 1$ then $|{\mathcal E}_1| \geq n-1$. $D$ can achieves utility value $\tau - (n-1) c_D$ by playing a tree network. Since $\mathbb{1}_{E_1\backslash E_A} \neq \mathbb{1}_{E_1}$ then  $|{\mathcal E}_A| \geq 1$ and $U_A \leq 1-\tau - c_A$. The bound is achieved by attacking any one link created by $D$. We further can show that $A$ needs to attack $k+1$ links such that $D$ will not heal the network. 
\end{proof}

\textit{Example:} In regime 2, for SPEs in Situation 5, the network remains empty since $D$ does not protect nor heal. An illustration of SPE in Situation 3 with $k = 1$ is depicted in Fig. \ref{SPE_ex2_2}. Specifically, $D$ creates a connected network with tree structure initially. Then, $A$ compromises any $k+1=2$ links to disconnect the network. Since $D$ is willing to recover at most $k=1$ link, $D$ does not heal the network at time $\tau+\tau_R$.

\begin{figure}[!t]
\begin{centering}
\includegraphics[width=0.9\columnwidth]{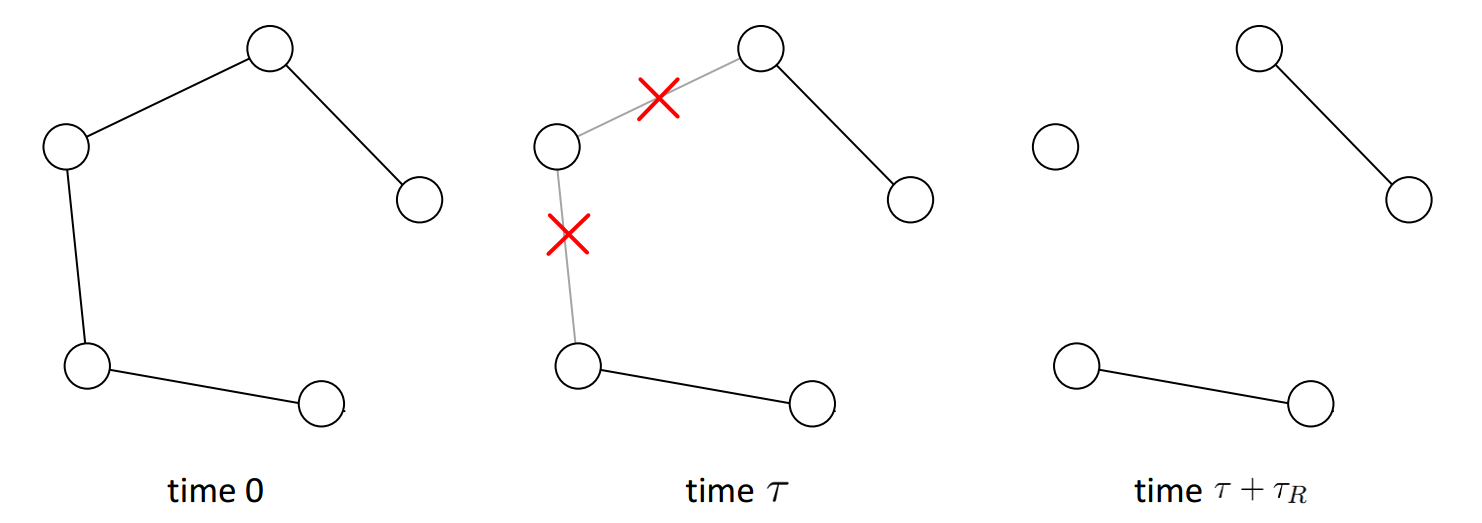} % double-c: 1; single-c:0.7
\par\end{centering}
\caption{\label{SPE_ex2_2}
The SPE lies in Situation 3 with $k=1$. Thus, $D$ will only create a tree network followed by $A$ compromising any 2 links to disconnect the network, and $D$ does not recover at time $\tau+\tau_R$.}
\end{figure}

In the following, we focus on the SPEs in Situations $1$ and $2$. In both cases, $\mathbb{1}_{E_1}=1$. Thus, $D$ creates a connected network initially. For each node $i \in {\mathcal N}$, let $d_i$ be its degree.  To facilitate the analysis, we focus on the potential best response strategies of $A$ to $E_1$ which are summarized in the following three distinct cases:
\begin{enumerate}[label=(\roman*)]
\item $A$ does not attack and obtains a utility of $U_A^{(1)} = 0$;
\item $A$ attacks sufficiently many links so that the network admits $2$ components, i.e., $A$ attacks exactly $\min_{1 \leq i \leq n}d_i$ links to disconnect a node of minimal degree. Then, $D$ heals the network by constructing $1$ link, and $A$ receives utility 
\begin{equation}U_A^{(2)} = \tau_R- (\min_{1 \leq i \leq n}d_i) c_A.\label{eq:uii}
\end{equation}
\item $A$ attacks sufficiently many links so that the network admits $\ell+2$ components, for some sufficiently large $\ell$ (whose exact value is discussed in the following two lemmas). Then, $D$ does not heal the network, and $A$ receives utility
\begin{equation}
U_A^{(3)} = 1-\tau- |{\mathcal E}_A| c_A.\label{eq:uiii}
\end{equation}
 Note that any intermediate value of components in the range $\llbracket 2 ; \ell+2 \rrbracket$ cannot happen at SPE since it amounts to a lower utility for $A$. The current case (iii) belongs to Situation 3 which eases the analysis in Lemmas \ref{lem:sit2} and \ref{lem:sit1}.
\end{enumerate}

The next lemma characterizes the SPEs for Situation 2.

\begin{lemma}\label{lem:sit2}
The only SPEs in Situation $2$ are such that $|{\mathcal E}_1|=n-1$, $|{\mathcal E}_A|=1$, $|{\mathcal E}_2|=1$,  $U_D = 1-\tau_R-n c_D$, and $U_A=\tau_R-c_A $. 
Furthermore, it occurs only if $c_A \leq \tau_R$ and 
$\floor*{\frac{1-\tau-\tau_R}{c_D}} > \floor*{\frac{1-\tau-\tau_R}{c_A}}$.
\end{lemma}

\begin{proof}
At an SPE in Situation $2$, the utility of $D$ is of the form $1-\tau_R - c_D (|{\mathcal E}_1| + |{\mathcal E}_2|)$. Then, it is a best strategy for $D$ to heal the network at time $\tau+\tau_R$, i.e., $1-\tau_R - (|{\mathcal E}_1| + |{\mathcal E}_2|) c_D \geq \tau -|{\mathcal E}_1| c_D.$
Thus, $|{\mathcal E}_2| \leq \floor*{\frac{1-\tau - \tau_R}{c_D}} = k$, and $k$ is the maximum number of links that $D$ can create at time $\tau+\tau_R$ at an SPE. In addition, at this SPE, $D$ receives a higher reward than by using its best strategy in Situation $3$, i.e., $1-\tau_R - (|{\mathcal E}_1| + |{\mathcal E}_2|) c_D \geq \tau-(n-1) c_D.$
Thus, $|{\mathcal E}_1| + |{\mathcal E}_2| \leq \floor*{\frac{1-\tau - \tau_R}{c_D}} +(n-1)=k+(n-1).$ Since $k < n-1$, then altogether $D$ can create at most $|{\mathcal E}_1| + |{\mathcal E}_2| \leq 2(n-1)$ links.

For any SPE in Situation $2$, note that $|{\mathcal E}_1| \geq n-1$. Thus, we can write $|{\mathcal E}_1| = n-1+\alpha$ and $|{\mathcal E}_2| \leq k - \alpha$, for some $\alpha < k$. 
%Consider the network created at time $0$ by the Designer and denote by $d_i$ the degree of node $i$ for $1 \leq i \leq n$. Since there are at most $2(n-1)$ links then $1 \leq d_i \leq 3$. 
For Situation $2$, we obtain $U_A^{(2)} \geq U_A^{(1)}$ which yields $(\min_{1 \leq i \leq n}d_i) \leq \floor*{\frac{\tau_R}{c_A}}$. If $\tau_R < c_A$, then no SPE exists in Situation $2$. 
Further, based on 
$0 \leq U_A^{(2)} - U_A^{(3)} = (|{\mathcal E}_A| - (\min_{1 \leq i \leq n}d_i))  c_A - (1-\tau-\tau_R)$, we obtain
$|{\mathcal E}_A| \geq \ceil*{\frac{1-\tau-\tau_R}{c_A}} + (\min_{1 \leq i \leq n}d_i)$.
Since at $\tau+\tau_R$, $D$ can create at most $k-\alpha$ links, then the goal of $A$ in case (iii) is to create at least $\ell = k-\alpha+2$ components in the network (i.e., to create a $k-\alpha+1$ cut). Hence, $D$ constructs ${\mathcal E}_1$ in a way that at least $k_A + (\min_{1 \leq i \leq n}d_i)$ links need to be removed so that the network consists of $k+2-\alpha$ components, where $k_A:=\ceil*{\frac{1-\tau-\tau_R}{c_A}}$. 

Recall that $k= \floor*{\frac{1-\tau - \tau_R}{c_D}}$ is the maximal number of links that $D$ can recover at phase $\tau+\tau_R$.
Suppose that $k < k_A$ (i.e., $k \leq k_A-1$). Then, for any $E_1$, consider the following attack: first remove $\alpha$ links so that the resulting network is a tree and then remove $k_2+1-\alpha$ links. Then, the resulting network has exactly $n-2-k+\alpha$ links, i.e., it has $n-(n-2-k+\alpha) = k-\alpha+2$ components and is obtained using $k+1 < k_A + (\min_{1 \leq i \leq n}d_i)$ links. Thus, if $k < k_A$, no SPE in Situation $2$ exists.
If $k > k_A+1$ (i.e., $k \geq k_A$), then we consider the strategy that $D$ creates a line network at time $0$. Then to induce $k+2$ components, $A$ needs to remove $k+1$ links. However, due to $k > k_A+1$, it is not of the best interest to $A$. Instead, the best response for $A$ is to attack exactly one link (one being adjacent to one of the nodes with degree $1$). Then, the best strategy for $D$ is to re-create this compromised link at time $\tau+\tau_R$ which is an SPE. It is strategic as it minimizes the number of created links. 
 \end{proof}
 
In Lemma \ref{lem:sit2}, the condition $c_A\leq \tau_R$ ensures that $A$ has an incentive to compromise the network, and the condition $\floor*{\frac{1-\tau-\tau_R}{c_D}} > \floor*{\frac{1-\tau-\tau_R}{c_A}}$ guarantees that $D$ is capable to heal the network after the attack.  Note that when these two conditions are satisfied, all other strategies that $D$ creates a tree network at phase 0 and $A$ attacks one link which is further reconnected by $D$ also constitute SPEs of Situation 2. 

To study the SPE in Situation 1, for convenience, we  denote
$$ k_A^R \equiv \floor*{\frac{\tau_R}{c_A}}\ \mathrm{and}\ k_A^H \equiv \floor*{\frac{1-\tau}{c_A}},$$
where $k_A^R$ (resp. $k_A^H$) corresponds to the maximal number of attacks that $A$ is willing to deploy to disconnect the network during the phase interval $[\tau,\tau+\tau_R]$ (resp. $[\tau,1]$) so that $U_A^{(2)}$ (resp. $U_A^{(3)}$) achieves a positive value. 

The following lemma characterizes the possible SPEs in Situation 1.
\begin{lemma}\label{lem:sit1}
If $\tau_R/c_A > n-1$ or $\floor*{\frac{1-\tau}{c_A}} >  \floor*{\frac{1-\tau}{c_D}}$, then no SPE exists in Situation $1$. 
Otherwise, let
\begin{equation}
\delta = \left\{ 
\begin{array}{ll}
\ceil*{\frac{n \left( k_A^R+1\right)}{2}} & \text{if } k \geq 1 \text{ and } k_A^R  > 1, \\
\ceil*{\frac{n \left( k_A^H+1\right)}{2}} & \text{if } k = 0 \text{ and } k_A^R  > 1, \\
n & \text{if } k_A^H = k+1 \text{ and } k_A^R  = 1, \\
n+\floor*{\frac{n}{k}}+\ceil*{\frac{\floor*{\frac{n}{k}}}{2}} & \text{if } k_A^H \neq k+1 \text{ and } k_A^R  = 1, \\
n-1 & \text{if } k_A^H = k \text{ and } k_A^R  = 0, \\
n & \text{if } k_A^H \neq k \text{ and } k_A^R  = 0. 
\end{array} \right. \label{eq:delta}
\end{equation}
If $1 < \delta c_D$ or if $1-\tau < (\delta-n+1) c_D$, then no SPE in Situation $1$ exists.
Otherwise, the unique SPE is such that $U_D = 1 - \delta c_D$ and $U_A = 0$.
\end{lemma}

\begin{proof}
See Appendix \ref{appLemma}.
\end{proof}

\textit{Example:} For clarity, an illustration of SPE in Situation 1 with $\delta = 5$ is depicted in Fig. \ref{SPE_ex2_1}. There are 5 nodes in the network and the parameters are $k_A^H = 2$ and $k_A^R = 1$. Specifically, $D$ creates a 2-Harary network with the ring topology initially. Then, $A$ is not capable to attack. The network remains connected over the entire time period.

\begin{figure}[!t]
\begin{centering}
\includegraphics[width=0.9\columnwidth]{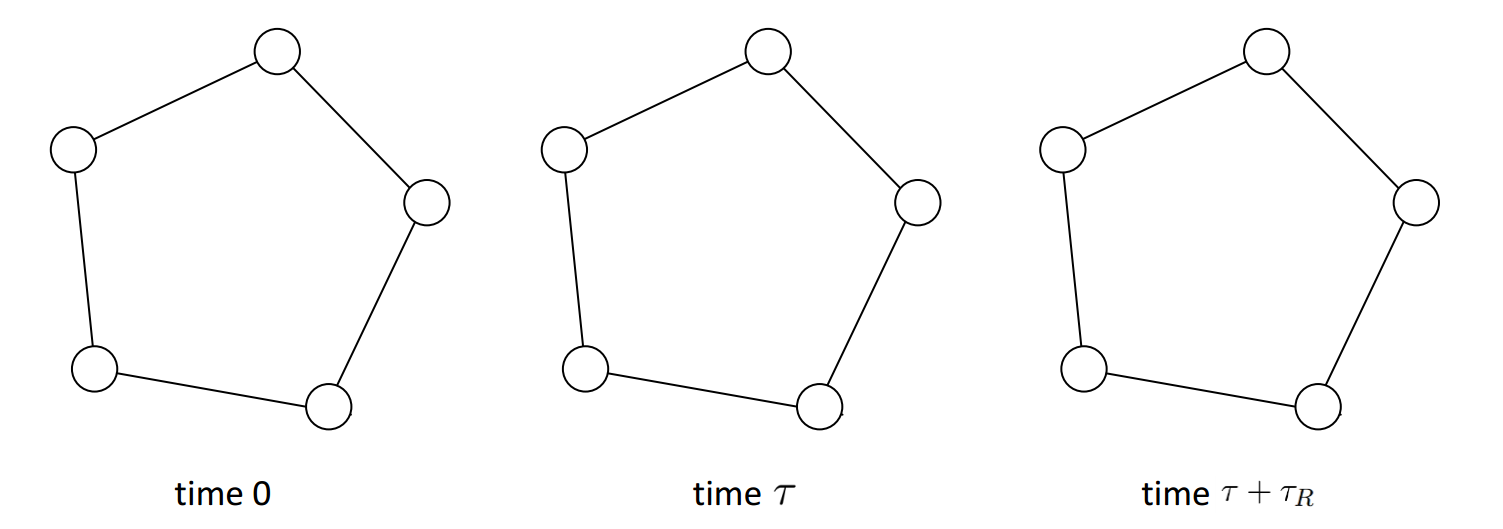} % double-c: 1; single-c:0.7
\par\end{centering}
\caption{\label{SPE_ex2_1}
The SPE lies in Situation 1 with $\delta=5$ ($k_A^H = 2$ and $k_A^R = 1$). Thus, $D$ creates a 2-Harary network with the ring topology. $A$ will not attack and thus $D$ does not heal the network.}
\end{figure}

For convenience, the results of Lemmas \ref{lem:penible}, \ref{lem:sit2} and \ref{lem:sit1} are summarized in Table~\ref{tab:SPE2}.
\begin{center}
\begin{table}[h]
{\small
%\hfill{}
$$\hspace{-1em}\begin{array}{|@{}c@{\,}|@{}c@{\;}c@{\;}c@{\,}|@{\,}c@{\;\;}c@{}|}
\text{Situation} & |{\mathcal E}_1| & |{\mathcal E}_A| & |{\mathcal E}_2| & U_D & U_A \\
\hline
1 & \delta & 0 & 0 & 1-c_D \delta & 0\\
2 & n-1 & 1 & 1 & 1- \tau_R-n c_D & \tau_R-c_A\\
3 & n-1 & k+1 & 0 & \tau-(n-1)c_D & 1-\tau-(k+1)c_A\\
5 & 0 & 0 & 0 & 0 & 1\\
\hline
\end{array}$$
}
\hfill{}
\caption{Different potential SPEs when $1-\tau-\tau_R < (n-1)c_D$ (Note: $\delta$ is given by Eq.~\eqref{eq:delta}, and $ k = \floor*{\frac{1-\tau-\tau_R}{c_D}}$).}
\label{tab:SPE2}
\end{table}
\end{center}

We next comment on the strategies of $D$ and $A$ at SPEs. Specifically, the players' strategies in Situation 2 under regime 2 are the same as the corresponding ones under regime 1. In Situation 3, $D$ creates a tree network at time 0 and does not heal it after $A$ compromising any $k+1$ links at phase $\tau$. Depending on the system parameters, in Situation 1, $D$ creates a connected network using $\delta$ links either in a tree, ring or Harary network topology, and $A$ does not attack.

\textit{Remark:} In the previous two Sections \ref{sec:1} and \ref{sec:2}, we have not explicitly determined those SPEs satisfying the boundary conditions. Note that at  boundaries where multiple SPEs could be feasible, the defender playing a leader role will first choose the one that yields the highest utility. Then, after fixing the defender's strategy, the attacker selects the SPE that maximizes its payoff.  

\subsection{Discussions on Constrained Action Set of $A$}
In some scenarios, $A$ may not be capable to attack a particular set of links due to constraints. Thus, some links  initially created by $D$ cannot be compromised by $A$, and they can be regarded as \textit{secure links}. The major SPE analysis of this paper is still valid for this constrained scenario with extra considerations on $A$'s feasible action set. We present the results for this extension in regime 1 briefly as follows, and the results in regime 2 can be obtained using similar arguments.

First, we consider the case that every node can create at least one secure link with other nodes. Then the SPE in Situation 1 under $k_A^R>0$ becomes as $|\mathcal{E}_1| = n-1$, $|\mathcal{E}_A| = 0$, and $|\mathcal{E}_2| = 0$. In this subcase, $D$ can create a connected network with all secure links using a tree topology and thus Harary network, $|\mathcal{E}_1| = \ceil*{\frac{n(k_A^R+1)}{2}}$, is not optimal to $D$.  Furthermore, Situation 2 is not possible as the network created by $D$ cannot be attacked. In addition, Situation 4 remains the same in this case. We next investigate cases in Situation 2. Indeed, SPE in  Situation 2 occurs if there exists at least a single link in the tree network created by $D$ at phase 0 which is insecure. Then, $A$ disconnects the network by compromising this vulnerable link. Finally, we analyze the case when a subset of nodes in the network can form secure links with others. In this scenario, the results of Situation 1 $\&\ k_A^R=0$, Situation 2, and Situation 4 in Table \ref{tab:SPE} still hold. For Situation 1 $\&\ k_A^R>0$, $D$ does not need to create a Harary network at phase 0 as some created links are secure. To this end, we can leverage network contraction \cite{chen2019} to derive the SPE. Network contraction refers to the principle that if there is a secure link between two nodes, we can aggregate them together and see them as a single super node.  In Situation 1 $\&\ k_A^R>0$, depending on the places where secure links can be formed, it leads to different policies for $D$ at phase 0. We illustrate the design principle for Situation 1 $\&\ k_A^R>0$ in Fig. \ref{network_contraction}. In this example, $|\mathcal{E}_1|=5$ is sufficient in the constrained scenario for $D$ to construct a secure network at time $0$, while it requires $|\mathcal{E}_1|=6$ links in the unconstrained counterpart.

\begin{figure}[t]
\centering
\includegraphics[width=1\columnwidth]{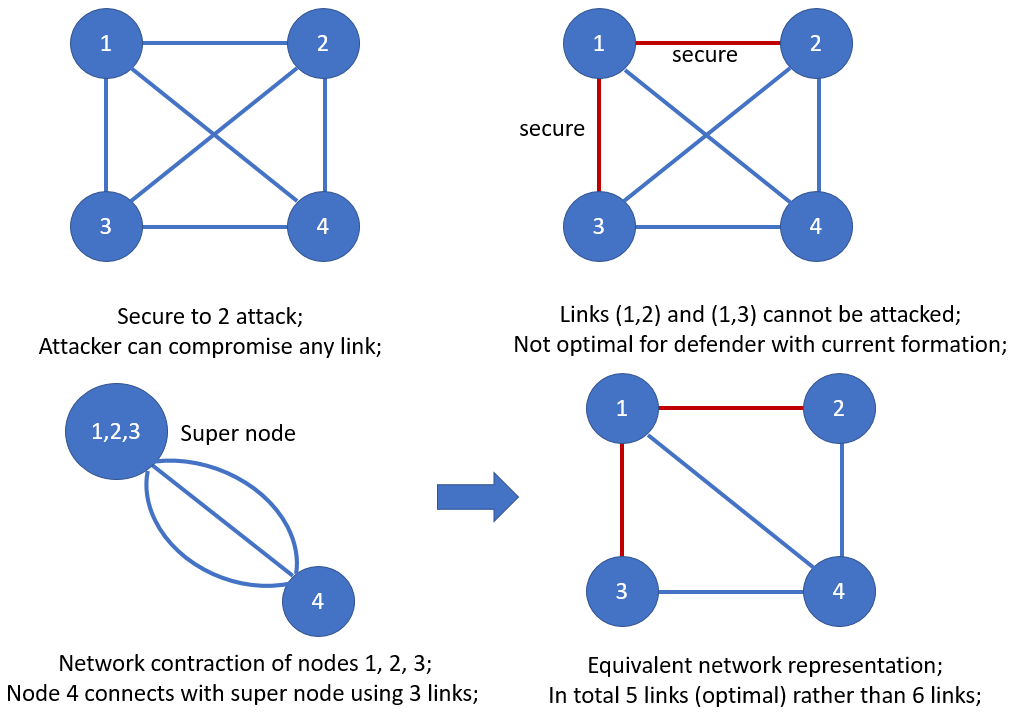}
\caption{Illustration of network contraction for designing $D$'s optimal strategy when a subset of nodes can form secure links with others. In the example, 6 links are required for the network being resistant to 2 link removals if $A$ can compromise any link.  When links (1,2) and (1,3) cannot be attacked,  nodes 1, 2, and 3 can be aggregrated as a super node by network contraction. Then, node 4 connects with the super node using 3 links. In sum, 5 links are sufficient for this constrained scenario which is different from the unconstrained case.}\label{network_contraction}
\end{figure}

\section{Network Resilience and Strategic Attack}\label{resilience}
In this section, we investigate the impact of network resilience on the SPE of the dynamic game and the attacker's behavior on the timing of attack.

\subsection{Resilience Planning}\label{resilience_analysis}
The infrastructure network resilience is measured by the response and recovery time after the cyber attack which is $\tau_R$ in our scenario. Thus, instead of merely maximizing $U_D$, the network operator should also take resilience metric $\tau_R$ into account. Thus, the aggregated objective function of $D$ can be formulated as follows:
\begin{equation}
F_D({\mathcal E}_1, {\mathcal E}_2, {\mathcal E}_A,\tau_R) = U_D({\mathcal E}_1, {\mathcal E}_2, {\mathcal E}_A) - R_D(\tau_R),
\end{equation}
where $R_D: [0,1]\rightarrow [0,1]$ quantifies the normalized system resilience cost. Specifically, $R_D$ is a monotonically decreasing function with respect to $\tau_R$. By considering the SPE of the dynamic game, $D$ chooses the best $\tau_R$ that results in an optimal utility $F_D$.

Based on Section \ref{SPE_analysis}, we obtain the following results. In regime 1 with agile resilience, i.e., $\tau_R<1-\tau-(n-1) c_D$, the utilities of $D$ under various SPE are summarized in Table \ref{tab:regime1_F_D}.

\begin{center}
\begin{table}[h]
{\small
%\hfill{}
$$\hspace{-1em}\begin{array}{|@{}c@{\,}|@{\,}c@{\;\;}|}
\mathrm{Situation} & F_D \\
\hline
1 \& k_A^R>0 &  1-c_D \ceil*{\frac{n(k_A^R+1)}{2}}-R_D(\tau_R)\\
1 \& k_A^R=0 & 1-(n-1) c_D-R_D(\tau_R)\\
2  & 1- \tau_R-n c_D-R_D(\tau_R)\\
4  & 1-\tau-\tau_R-(n-1)c_D-R_D(\tau_R)\\
\hline
\end{array}$$
}
\hfill{}
\caption{Utilities of $D$ under different potential SPE when $1-\tau-\tau_R > (n-1)c_D$ (Note: $k_A^R = \floor*{\frac{\tau_R}{c_A}}$).}
\label{tab:regime1_F_D}
\end{table}
\end{center}

Similarly, in regime 2 with $\tau_R>1-\tau-(n-1) c_D$, $D$'s utilities with different scenarios are presented in Table \ref{tab:regime2_F_D}.

\begin{center}
\begin{table}[h]
{\small
%\hfill{}
$$\hspace{-1em}\begin{array}{|@{}c@{\,}|@{\,}c|}
\text{Situation} & F_D \\
\hline
1 & 1-c_D \delta-R_D(\tau_R)  \\
2 & 1- \tau_R-n c_D-R_D(\tau_R) \\
3 & \tau-(n-1)c_D-R_D(\tau_R) \\
5 & 0\\
\hline
\end{array}$$
}
\hfill{}
\caption{Utilities of $D$ under different potential SPE when $1-\tau-\tau_R < (n-1)c_D$ (Note: $\delta$ is given by Eq.~\eqref{eq:delta}).}
\label{tab:regime2_F_D}
\end{table}
\end{center}

\textit{Remark:} Under different regimes and situations, the aggregated payoff $F_D$ of $D$ admits various forms. Comparing the values of $F_D$ in Tables \ref{tab:regime1_F_D} and \ref{tab:regime2_F_D}, the designer selects a $\tau_R$ that yields the largest $F_D$, and the corresponding SPE strategies can be determined based on Tables \ref{tab:SPE} and \ref{tab:SPE2}.

\subsection{Strategic Timing of Attack}
The attacker's behavior depends on the recovery ability of the network. When $A$ decides to compromise the network, then choosing the attacking phase $\tau$ also becomes a critical issue. Specifically, for a given $\tau_R$, $A$ needs to decide the value of $\tau$. As shown in Lemma \ref{lem:situations}, $A$ compromises the network only if $D$ creates a connected network initially. Thus, we focus on two Situations: 2 and 3. Proposition \ref{prop:1} indicates that when Situation 2 is an SPE, the corresponding utility of $A$ is $U_A=\tau_R-c_A$ which does not depend on the attacking phase $\tau$. In an SPE of Situation 3, $D$ does not heal the network after attack, and the utility of $A$ is $U_A=  1-\tau-(k+1)c_A$. Hence, the timing of attack $\tau$ has an influence on $A$'s payoff. In another case when SPE takes a form of Situation 4, $A$'s utility is $\tau+\tau_R$ which is also influenced by the attacking phase. Despite that $A$ does not attack, its action induces a threat to the network. We summarize the results in the following Lemma.

\begin{lemma}\label{lem:attack_time}
When SPE of the game admits a form of Situation 3, then the best timing of attack for $A$ is  to choose the smallest $\tau$ in the set $\{\tau\big|\tau\geq \frac{1-\tau_R}{(n-1)c_D},\ \floor*{\frac{1-\tau}{c_A}}\geq\floor*{\frac{1-\tau-\tau_R}{c_D}}+1\}$. When SPE takes a form of Situation 4, then the best $\tau$ for $A$  is choosing the largest value in the set $\{ c_D, 1-\tau_R-(n-1)c_D, c_D \ceil*{\frac{n(k_A^R-1)}{2}+1} -\tau_R \}$. When SPE of the game is of another form except for Situations 3 and 4, then $\tau$ does not affect the utility of $A$.
\end{lemma}

\begin{proof}
The attacker chooses a $\tau$ to maximize its utility $1-\tau-(\floor*{\frac{1-\tau-\tau_R}{c_D}}+1)c_A$ while satisfying the conditions $\floor{\frac{1-\tau}{c_A}}\geq \floor*{\frac{1-\tau-\tau_R}{c_D}}+1$ and $1-\tau-\tau_R<(n-1)c_D$. The objective function indicates that a smaller $\tau$ yields a higher payoff of $A$. Thus, the best timing of attack is the smallest $\tau$ resulting in an SPE of Situation 3. We relax the strict inequality constraint by including the boundary, since when $\tau= \frac{1-\tau_R}{(n-1)c_D}$, $D$ does not heal the network and Situation 3 is still an SPE. Similarly,  in Situation 4, those boundary values of $\tau$ at the inequality constraint are feasible since $D$ chooses not to create a connected network if the payoffs are the same. 
\end{proof}

In Situation 3, $A$ prefers to attack the network in an early phase which aligns with the fact that $D$ does not recover the network, and hence $A$ receives the total rewards after $\tau$. In contrast, $A$ chooses to compromise the network at a larger phase $\tau$ in Situation 4 (though he does not really attack since the network is not connected), which extracts all the utility from time 0 to $\tau+\tau_R$.

\section{Case Studies}\label{case_study}
In this section, we use case studies of UAV-enabled communication networks to corroborate the obtained results. UAVs become an emerging technology to serve as communication relays, especially in disaster recovery scenarios in which the existing communication infrastructures are out of service \cite{tuna2014unmanned}.
In the following, we consider a team of $n=10$ UAVs. The normalized unitary costs of creating and compromising a communication link between UAVs for the operator/defender and adversary are $c_D = 1/20$ and $c_A = 1/8$, respectively.

\begin{figure}[t]
\centering
\includegraphics[width=0.85\columnwidth]{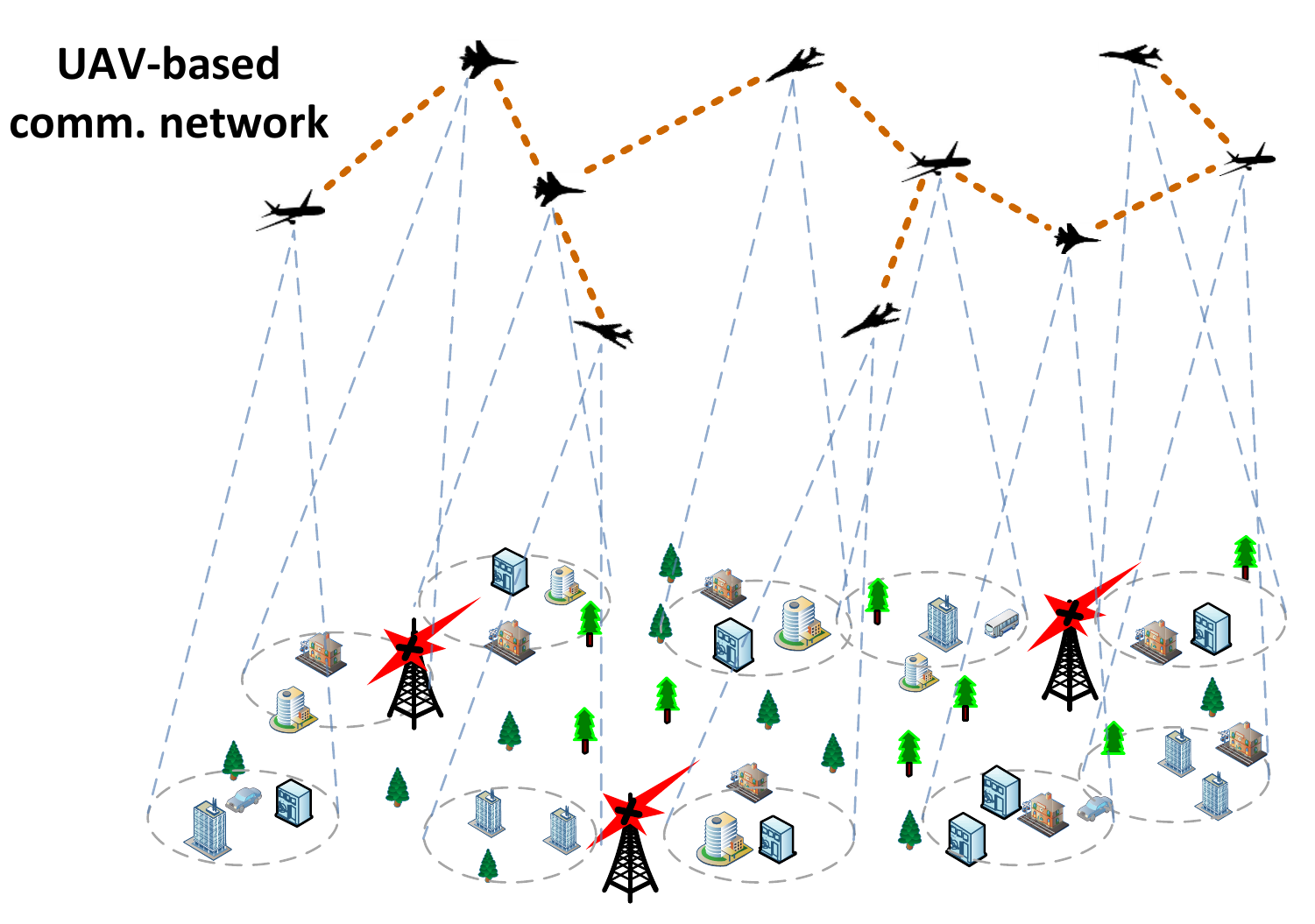}
\caption{UAV-enabled communication networks for disaster recovery. The UAVs form a tree network at SPE ($\tau=0.3$, $\tau_R=0.2$).}\label{UAV}
\end{figure}

\subsection{Illustrations of SPEs (Results in Section \ref{SPE_analysis})}
First, we illustrate SPE of the game when network resilience cost and timing of attack are not considered (results in Section \ref{SPE_analysis}). Specifically, the adversary attacks the network at phase $\tau = 0.3$, and the defender heals it after $\tau_R = 0.2$. The UAV-enabled communication network configuration at SPE is shown in Fig. \ref{UAV} which admits a tree structure, and $A$ does not attack the network at SPE. In addition, the utilities for $D$ and $A$ at SPE with $\tau_R\in[0, 0.6]$ are shown in Fig. \ref{utility_SPE}. The SPE encounters switching with different $\tau_R$. As $\tau_R$ increases, the UAV network operator needs to allocate more link resources to secure the network. Otherwise, the attacker has an incentive to compromise the communication links with a positive payoff. Specifically, when $\tau_R<0.375$, $A$ does not attack the UAV network, and $D$ obtains a positive utility by constructing a securely connected network. The secure network admits various structures depending on $\tau_R$. As shown in Fig. \ref{utility_SPE}, it can be in a tree network or a Harary network and the SPEs are in Situation 1. When $0.375<\tau_R<0.5$, the defender creates a connected network with the minimum effort, i.e., $9$ links, at phase $0$. In this interval, the attacker will successfully compromise the system during phase $[\tau,\tau+\tau_R]$, and the defender heals the network afterward. The initially connected network in this regime admits a tree structure, and it may not be the same as the one created in the regime of $\tau_R<0.375$. When $\tau_R$ exceeds $0.5$, the defender does not either protect or heal the network. The reason is that a larger $\tau_R$ provides more incentives for the attacker to compromise the links and receive a higher payoff. Furthermore, the aggregated utility for the defender from two intervals, i.e., from the initial phase to the attacking phase and from the recovery phase to the terminal phase, is small, and hence it does not provide sufficient incentive for the defender to protect and recover the network. This also indicates that agile resilience is critical in mitigating cyber threats in the infrastructure networks.

\begin{figure}[t]
\centering
\includegraphics[width=0.95\columnwidth]{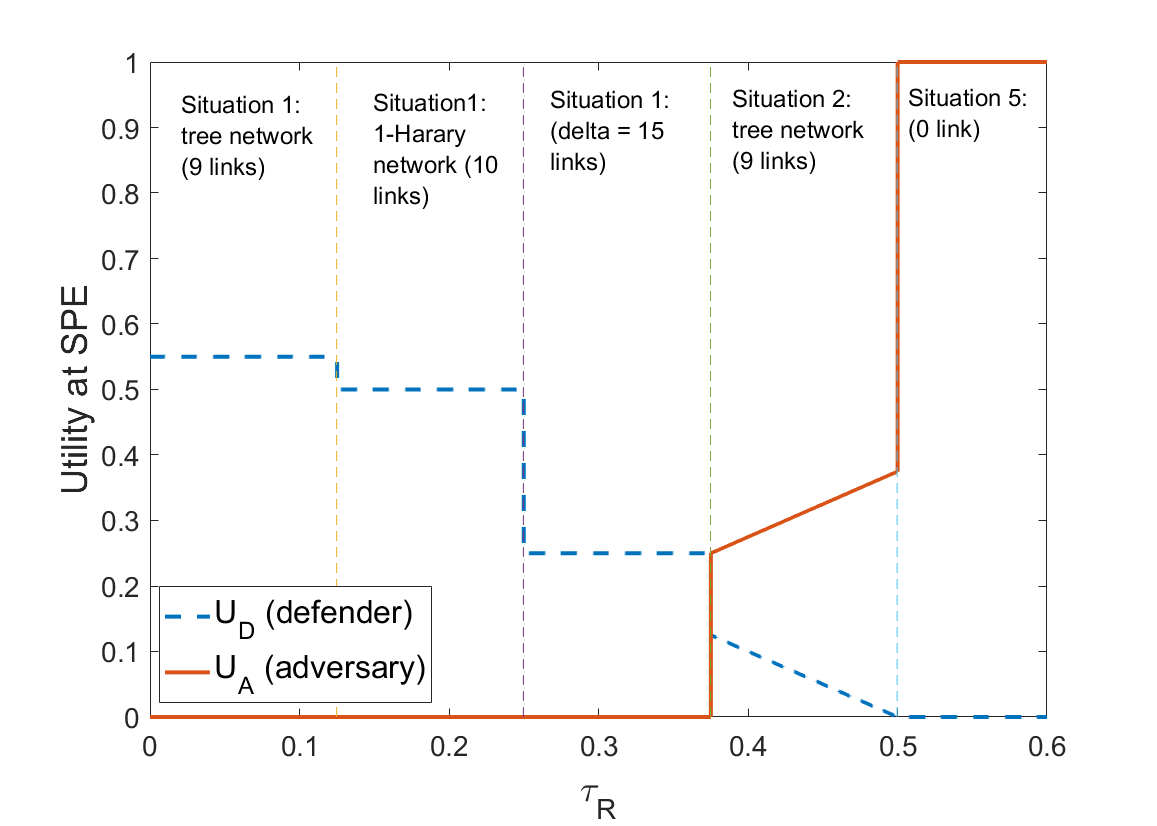}
\caption{Utilities for $D$ and $A$ at SPE with varying $\tau_R$. The SPEs and the strategies of $D$ and $A$ are different with the increase of $\tau_R$.}\label{utility_SPE}
\end{figure} 

\subsection{Strategic Resilience Planning}
Next, we take into account the cost of network resilience and study its impact on the SPE. The cost function of resilience is $R_D(\tau_R) = (\tau_R-1)^4$. The convexity of $R_D$ indicates that the marginal cost of resilience increases as $\tau_R$ decreases. The timing of attack is fixed to $\tau=0.3$ in this case study. The equilibrium strategies of both players under costly network resilience are illustrated in Fig. \ref{resilience_SPE}. Based on the analysis in Section \ref{resilience_analysis}, $D$ chooses a $\tau_R$ that maximizes the net utility $F_D$. Though $U_D$ is larger in a regime with smaller values of $\tau_R$, the cost of agile network resilience is much higher for it being the best strategy of designer. In addition, the defender will not choose a $\tau_R$ in the intervals $[0,0.16]\cup[0.375,0.6]$ since $F_D$ is negative. Hence, the optimal resilience planning of $D$ is $\tau_R = 0.25$ which yields the optimal payoff $F_D = 0.183$. At this SPE, which falls into Situation 1, $D$ creates a $(10,1)$-Harary network using 10 links initially and $A$ does not attack.

\begin{figure}[t]
\centering
\includegraphics[width=0.9\columnwidth]{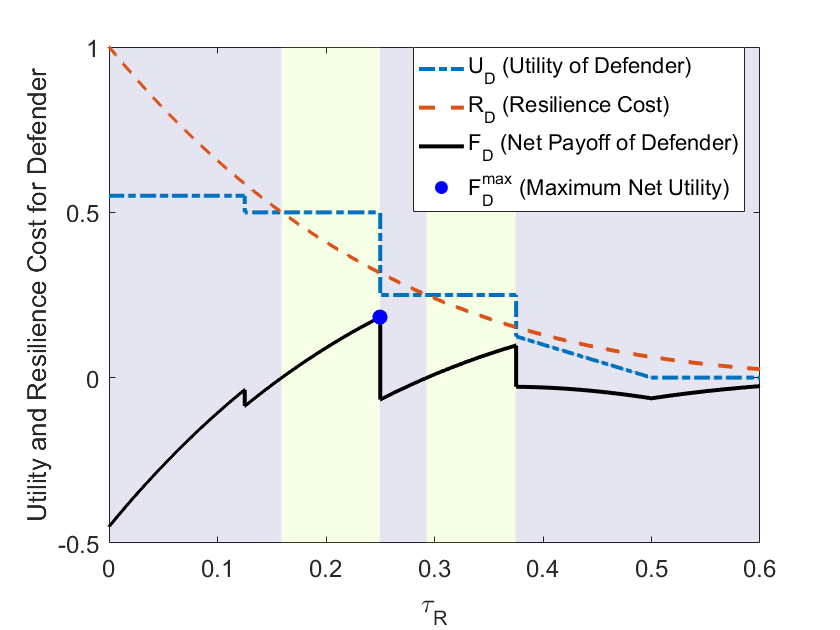}
\caption{Defender's utility with varying $\tau_R$ by considering the resilience cost. The optimal resilience planning is achieved at $\tau_R = 0.25$. Values of $\tau_R$ in the interval $[0,0.16]\cup[0.375,0.6]$ are not feasible since $F_D$ is negative.}\label{resilience_SPE}
\end{figure} 

\begin{figure}[t]
  \centering
  \subfigure[$\tau=0.5$]{
    \includegraphics[width=0.9\columnwidth]{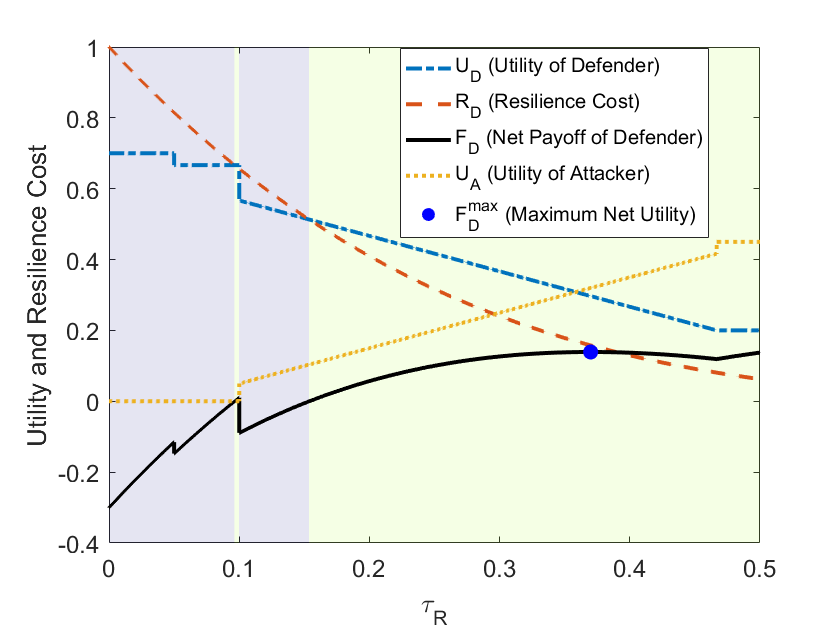}\label{attacking_tau5}}
	 \subfigure[$\tau=0.55$]{
    \includegraphics[width=0.9\columnwidth]{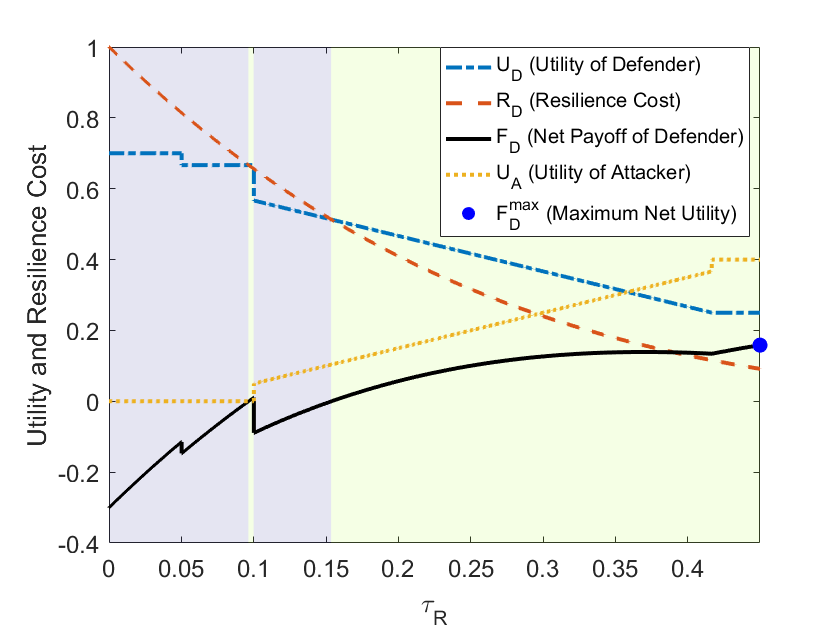}\label{attacking_tau55}}       
  \caption[]{(a) and (b) illustrate the strategies of $D$ and $A$ with different $\tau$. The SPE under optimal resilience planning switches from Situation 2 to Situation 3 as $\tau$ changes from 0.5 to 0.55. In both (a) and (b), the shaded grey areas, i.e., $\tau_R\in[0,0.096]\cup[0.01,0.154]$, are not feasible due to the negative value $F_D$.}
  \label{attacking_example}
\end{figure}

\subsection{Strategic Attacks and Resilience Planning}
We finally investigate the strategic attack behavior of $A$. In the following, the costs of creating and compromising a communication link are selected as $c_D = 1/30$ and $c_A = 1/20$, respectively. The SPEs and the corresponding utilities with $\tau = 0.5$ and $\tau=0.55$ are shown in Fig. \ref{attacking_example}. Specifically, Fig. \ref{attacking_tau5} shows that the optimal resilience planning for $D$ at $\tau = 0.5$ is $\tau_R=0.37$, leading to $F_D = 0.14$. Note that this SPE, where $D$ constructs a tree network at phase 0 and $A$ attacks one link at $\tau$ with $D$ healing the network afterward, belongs to Situation 2 in regime 2 as shown in Table 
\ref{tab:regime2_F_D}. As $\tau=0.55$, the best resilience planning of $D$ is to adopt $\tau_R = 0.45$ as illustrated in Fig. \ref{attacking_tau55}. At this SPE,  which is a case of Situation 3 in regime 2, $D$ creates a tree network initially and $A$ attacks one link at $\tau$, and $D$ does not recover the network. We can see that the SPE under optimal resilience planning switches from Situation 2 to Situation 3 as $\tau$ increases. In addition, the utility of $A$ varies under different SPEs. Based on Lemma \ref{lem:attack_time}, for an SPE in Situation 3, the attacker can increase its utility by choosing an appropriate $\tau$. Thus, we study the impact of attacking phase $\tau$ on the SPE of the game, and the result is depicted in Fig. \ref{optimal_attack_resilience}. Note that the utility of $D$ is optimal under each $\tau$ in the sense that the resilience cost $R_D$ is considered. When $\tau\in[0.4,0.515)$, the SPE belongs to Situation 2, and the optimal utilities of $D$ and $A$ remain as constants, where the resilience metric is given by $\tau_R = 0.37$. When $\tau\in[0.515,0.6]$, the SPE switches to a case of Situation 3. In this interval, $D$ does not recover after the attack and $\tau_R = 1-\tau$.
Furthermore, the optimal timing of attack is selected as $\tau = 0.515$, leading to $U_A = 0.435$, the largest utility of $A$. The result is in consistence with Lemma \ref{lem:attack_time}, indicating that a smaller $\tau$ yields a higher utility of $A$ when SPE admits a form of Situation 3. 

We next investigate the SPEs under the optimal resilience planning of $D$ and the strategic timing of attack of $A$ together over varying cost ratio $c_A/c_D$. We fix $c_D = 1/30$ and the ratio $c_A/c_D$ varies. Figure \ref{resilience_with_attacking_varying_cost} shows the obtained results. As the cost ratio $c_A/c_D$ increases, the utility of $A$ decreases monotonically. When $c_A/c_D\in [1, 2.2]$, the strategies of $D$ and $A$ does not change and the SPE belongs to Situation 3. Since the optimal $\tau_R$ and $\tau$ stay the same in this interval, the utility of $D$ remains unchanged. When $c_A/c_D\in [2.2, 2.4]$, the SPE switches to Situation 2 and $D$ heals the network after the attack. Since the recovery is agile ($\tau_R$ becomes smaller), $D$'s utility increases and $A$'s utility decreases dramatically in this interval. Furthermore, the strategic timing of attack $\tau$ varies to account for the better recovery speed $\tau_R$ and the increasing cost of attack.

\begin{figure}[t]
\centering
\includegraphics[width=0.9\columnwidth]{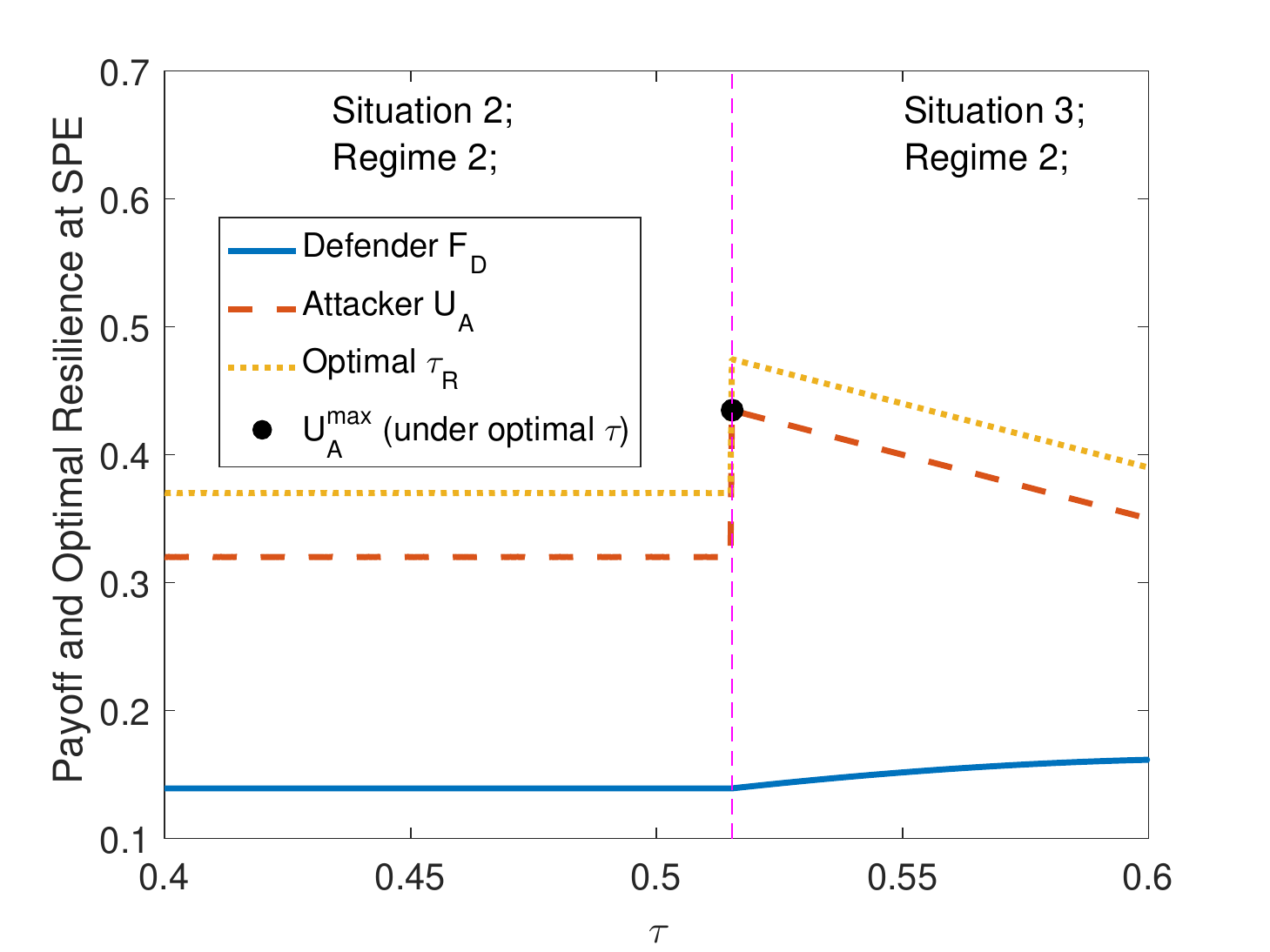}
\caption{Players' utilities at SPE with varying $\tau$ under the optimal resilience planning. The best timing of attack is $\tau = 0.515$ with optimal $\tau_R = 0.37$.}\label{optimal_attack_resilience}
\end{figure} 

\begin{figure}[t]
\centering
\includegraphics[width=0.9\columnwidth]{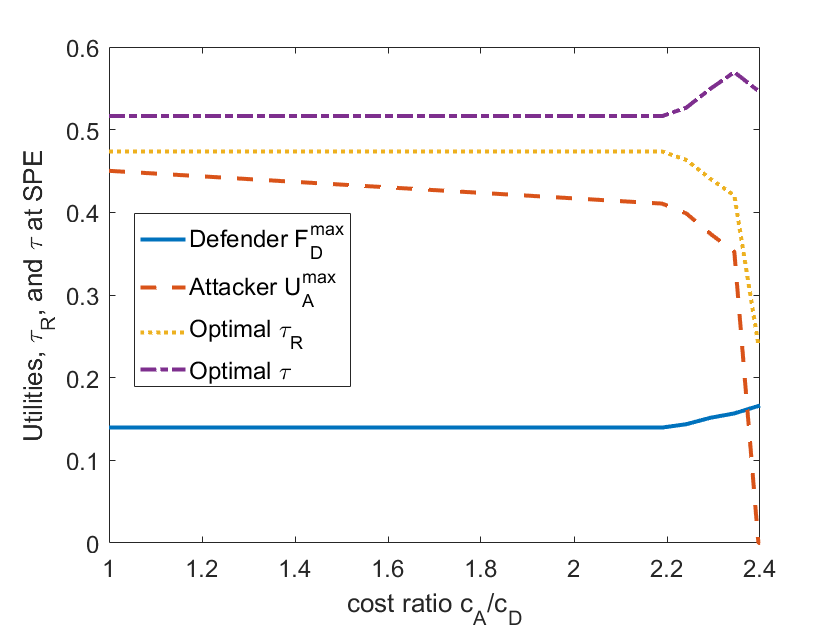}
\caption{Players' utilities, optimal resilience planing $\tau_R$, strategic timing of attack $\tau$, at SPE with varying $c_A/c_D$.}\label{resilience_with_attacking_varying_cost}
\end{figure}

\section{Conclusion}\label{conclusion}
In this paper, we have established a two-player three-stage dynamic game for the infrastructure network protection and recovery. We have characterized the strategic strategies of the network defender and the attacker by analyzing the subgame perfect equilibrium (SPE) of the game. With case studies on UAV-enabled communication networks for disaster recovery, we have observed that with an agile response to the attack, the defender can obtain a positive utility by creating a securely connected infrastructure network. Furthermore, a higher level resilience saves link resources for the defender and yields a better payoff. In addition, a longer duration between the attack and recovery phases induces a higher level of cyber threats to the infrastructures.  Future work would investigate dynamic games with incomplete information of the defender on the attacking time and attack cost.  Another direction is to design SPE strategies under the scenarios that the feasible action sets of both defender and attacker are constrained.

\appendices

\section{Proof of Lemma \ref{lem:sit1}}\label{appLemma}
\begin{proof}
%At an SPE in Situation $1$, $U_D$ is of the form  $1 - c_D |{\mathcal E}_1|$. Therefore, we obtain $ |{\mathcal E}_1| \leq \floor*{\frac{1}{c_D}}.$
%Further, $1 - c_D |{\mathcal E}_1|$ should be greater than the utility in Situation $3$, i.e., $\tau-(n-1) c_D$. Otherwise, $D$ will not create a secure network at $0$ which is resistant to attacks. Thus, we obtain $1 -\tau  \geq (|{\mathcal E}_1|-(n-1)) c_D$. Since $1-\tau-\tau_R < (n-1) c_D$, then $1 -\tau  \geq (|{\mathcal E}_1|-(n-1)) c_D + 1-\tau-\tau_R - (n-1) c_D$ which yields
%\begin{equation}
%\tau_R  \geq (|\mathcal E_1|-2(n-1)) c_D.
%\label{eq:4}
%\end{equation}
When $\tau_R/c_A > n-1$, $A$ always attacks the network at phase $\tau$, and hence Situation 1 is not possible.
The SPE in Situation 1 satisfies $U_A^{(1)} > U_A^{(2)}$ and $U_A^{(1)} > U_A^{(3)}$. 
Thus, the goal of $D$ is to create a network with the minimal cost such that all nodes have a degree of at least $\floor*{\frac{\tau_R}{c_A}}+1 = k_A^R+1$, and at least $\floor*{\frac{1-\tau}{c_A}}+1 = k_A^H+1$ links need to be removed to yield a network with $k+2$ components (i.e., the minimum $(k+1)$-cut requires at least $\floor*{\frac{1-\tau}{c_A}}+1$ links). 
For $k_A^R\geq 1$, we consider the strategy of $D$ that consists in creating an $(|{\mathcal N}|, k_A^R+1)$-Harary network. Thus, 
\begin{equation}
|{\mathcal E}_1| \geq
\left\{
\begin{array}{ll}
\ceil*{\frac{n \left( k_A^R+1\right)}{2}} &\text{if}\  k_A^R \geq 2, \\  
n &\text{if}\ k_A^R = 1,\\
n-1 &\text{otherwise.}\\
\end{array}
\right. \label{eq:5}
\end{equation}

Let $k_D^H \equiv \floor*{\frac{1-\tau}{c_D}}$. First,
suppose that $k_D^H < k_A^H$, and at phase $0$, $D$ constructs a network with $(n-1)+\overline{k}$ links for some $\overline{k} \leq k_D^H$. Consider the strategy for $A$ that consists in attacking randomly $k_A^H$ links. Since $k_A^H > k_D^H \geq \overline{k}$, then the resulting network has less than $n-1$ links and is thus disconnected. At phase $\tau+\tau_R$, $D$ can reconstruct at most $ (n-1)+k^H_D-(n-1)-\overline{k} = k^H_D-\overline{k}$ links. Then, the network at phase $\tau+\tau_R$ would contain at most $(n-1)+\overline{k}-k_A^H+k^H_D-\overline{k} = (n-1)+k^H_D-k_A^H < n-1$ links, and the network is disconnected. Therefore,  no SPE exists in Situation $1$ if $k_D^H < k_A^H$.

Conversely, suppose that $k_D^H \geq k_A^H$. Then, we have
%$
%k_A^H  \leq k_D^H\Rightarrow \frac{1-\tau}{c_A}-1 < \floor*{\frac{1-\tau}{c_A}} \leq 
% \floor*{\frac{1-\tau}{c_D}} \leq \frac{1-\tau}{c_D} 
%\Rightarrow \
%\frac{1}{c_A}-\frac{1}{c_D} < \frac{1}{1-\tau} \Rightarrow
%\frac{\tau_R}{c_A}-\frac{\tau_R}{c_D} < \frac{\tau_R}{1-\tau}<1
%$,
%which yields 
$k_A^R \leq \floor*{\frac{\tau_R}{c_D}}$.
Furthermore,
$
k_A^H  \leq k_D^H\Rightarrow
\frac{1}{c_A}-\frac{1}{c_D} < \frac{1}{1-\tau} \Rightarrow\ 
\frac{1-\tau-\tau_R}{c_A}-\frac{1-\tau-\tau_R}{c_D} < \frac{1-\tau-\tau_R}{1-\tau}<1
$,
which gives $\floor*{\frac{1-\tau-\tau_R}{c_A}} \leq k$.
Then, by definition, $k_A^H = \floor*{\frac{1-\tau}{c_A}} = \floor*{\frac{1-\tau-\tau_R}{c_A} + \frac{\tau_R}{c_A}} \leq \floor*{\frac{1-\tau-\tau_R}{c_A}} + \floor*{\frac{\tau_R}{c_A}} +1 \leq k + k_A^R+1$. Hence, we obtain $k_A^H \leq k + k_A^R +1$. Based on the obtained results, we next focus on four distinct cases and derive their corresponding SPEs.

\textbf{Case 1 ($k>0$ and $k_A^R> 1$):} If $k_A^R \geq 3$, then $k_A^R+1$ link removals are needed to disconnect the network, and any further additional component creation requires to remove at least $2$ links. Thus, at least $2k+k_A^R+1$ link removals are necessary so that the network has $k+2$ components. Then, based on $2k+k_A^R+1>k_A^H+1$, $A$ does not attack the network. 
If $k_A^R=2$, and if $k \leq \floor*{\frac{n}{2}}$, then at least $k_A^R+1+2k$ link removals are required, and otherwise (i.e., $k > \floor*{\frac{n}{2}}$) $k_A^R+1+k$ link removals are necessary. Thus, $A$ does not attack the network.

\textbf{Case 2 ($k=0$ and $k_A^R> 1$):} In this case, we have $k_A^H \leq k_A^R +1$. $A$ only needs to disconnect the network since $D$ does not heal due to $k=0$. Thus, if $k_A^R> 1$,  $D$ creates an $(|{\mathcal N}|, k_A^H+1)$-Harary network at phase 0.

\textbf{Case 3 ($k_A^R=0$):} In this case, if $k_A^H = k$, then $D$ creates a tree network which is an optimal strategy. Otherwise, $k_A^H = k+1$ in which case $D$ creates a ring network. 

\textbf{Case 4 ($k_A^R=1$):} In this scenario, if $k_A^H = k+1$, then the ring network, i.e., the $(|{\mathcal N}|, 2)$-Harary network, is optimal for $D$. Otherwise, if $k_A^H = k+2$, then $D$ needs to create a network of minimal cost such that no $k$ cut exists with $k+1$ links. To this end, we consider the following network. For each $i \in {\mathcal N}$, we create a link between nodes $i$ and $(i+1) \mod n$ (ring network). Then, we connect node $k$ to node $2k$, and connect node $2k$ to node $3k$, and so on. If $\floor*{\frac{n}{k}}$ is even, then we connect node $k\floor*{\frac{n}{k}}$ to node $0$. Otherwise, we connect node $0$ to any node of the network excluding $1$ and $n-1$.
Thus, the resulting network contains no $k$ cut of size $k+1$ links and is minimal in terms of the number of links. The resulting utility for $D$ is $U_D = 1-(n+\floor*{\frac{n}{k}}+\ceil*{\frac{\floor*{\frac{n}{k}}}{2}})c_D$. 

By defining $\delta$ as in \eqref{eq:delta}, the condition $1<\delta c_D$ ensures a positive utility for $D$ at SPE of Situation 1. The condition $1-\tau < (\delta-n+1) c_D$ guarantees that the SPE is achieved in Situation 1 instead of in Situation 3. 
\end{proof}

\bibliographystyle{IEEEtran}
\bibliography{IEEEabrv,ref}

\end{document}